\documentclass[12pt]{article}
\usepackage{amsmath}
\usepackage{graphicx}
\usepackage{natbib}
\usepackage{url} % not crucial - just used below for the URL 
\usepackage{hyperref}       % hyperlinks
\usepackage{url}            % simple URL typesetting
\usepackage{booktabs}       % professional-quality tables
\usepackage{amsfonts}       % blackboard math symbols
\usepackage{nicefrac}       % compact symbols for 1/2, etc.
\usepackage{xcolor}         % colors
\usepackage{xr}
\usepackage{amsfonts, amssymb, amsthm, natbib}

\usepackage{algorithm, algpseudocode, enumerate, hyperref, color, comment, bigints}
\usepackage{enumitem}
\usepackage{setspace}
\usepackage{dsfont}

\def\real{\mathbb R}
\def\Prob{\mathbb {P}}

\def\y{\pmb{y}}

\def\b{\pmb{b}}
\def\P{\mathcal{P}}
\def\A{\mathcal{A}}
\def\bbeta{\pmb{\beta}}
\def\bomega{\pmb{\omega}}

\newtheorem{thm}{Theorem}
\newtheorem*{thm*}{Theorem}
\newtheorem{lem}{Lemma}

\newtheorem{cor}{Corollary}
\newtheorem{prop}{Proposition}

\DeclareMathOperator*{\argmin}{arg\,min}

\newcommand{\snorm}[1]{\| #1 \|}

\def\bjk{\beta^M_{jk}}
\def\htjk{\widehat{\theta}^M_{jk}}
\def\hbjk{\widehat{\beta}^M_{jk}}
\def\tjk{\theta^M_{jk}}
\def\perpjk{\widehat{\mathcal{A}}^M_{jk}}
\def\interjk{\mathcal{I}_{jk}}

\def\Shat{\widehat{S}^M_{jk}}

\newcommand{\fromjacob}[1]{{\noindent\color{blue}From Jacob: #1}}

\newcommand{\fromyiling}[1]{{\noindent\color{purple}From Yiling: #1}}

\newcommand{\blind}{0}

\begin{document}

\def\spacingset#1{\renewcommand{\baselinestretch}%
{#1}\small\normalsize} \spacingset{1}

%%%%%%%%%%%%%%%%%%%%%%%%%%%%%%%%%%%%%%%%%%%%%%%%%%%%%%%%%%%%%%%%%%%%%%%%%%%%%%
\date{}
\if0\blind
{
  \title{\bf Reluctant Interaction Inference after Additive Modeling}
  \author{Yiling Huang \\
    Department of Statistics, University of Michigan\\
    Snigdha Panigrahi\thanks{
    The author gratefully acknowledges support from NSF CAREER Award DMS-2337882.}\hspace{.2cm}\\
    Department of Statistics, University of Michigan\\
    Guo Yu \\
    Department of Statistics and Applied Probability, \\University of California, Santa Barbara\\
    and \\
    Jacob Bien\\
    Department of Data Sciences and Operations,\\ 
    University of Southern California}
  \maketitle
} \fi

\if1\blind
{
  \bigskip
  \bigskip
  \bigskip
  \begin{center}
    {\LARGE\bf Reluctant Interaction Inference after Additive Modeling}
\end{center}
  \medskip
} \fi

\bigskip
\begin{abstract}
Additive models enjoy the flexibility of nonlinear models while still being readily understandable to humans. 
By contrast, other nonlinear models, which involve interactions between features, are not only harder to fit but also substantially more complicated to explain. 
Guided by the principle of parsimony, a data analyst therefore may naturally be reluctant to move beyond an additive model unless it is truly warranted. 

To put this principle of interaction reluctance into practice, we formulate the problem as a hypothesis test with a fitted sparse additive model (SPAM) serving as the null.
Because our hypotheses on interaction effects are formed after fitting a SPAM to the data, we adopt a selective inference approach to construct p-values that properly account for this data adaptivity.
%We adopt a selective inference approach to construct p-values for these data-adaptive hypotheses on interaction effects, which are formed only after fitting a SPAM to the given data.
Our approach makes use of external randomization to obtain the  distribution of test statistics conditional on the SPAM fit, allowing us to derive valid p-values, corrected for the over-optimism introduced by the data-adaptive process prior to the test.
Through experiments on simulated and real data, we illustrate that---even with small amounts of external randomization---this rigorous modeling approach enjoys considerable advantages over naive methods and %other selective inference methods like 
data splitting. %, recovering more accurate nonlinear models with interactions.
\end{abstract}

\noindent%
%{\it Keywords:}  3 to 6 keywords, that do not appear in the title
\vfill

\newpage
\spacingset{1.75} % DON'T change the spacing!

\section{Introduction}
\label{sec:intro}
The past decade has seen the widespread adoption of highly nonlinear models such as neural networks.  This trend has been accompanied by concerns about the reliability of such highly flexible models, when deployed in real-world settings \citep{damour2022underspecification}. In particular, \cite{fisher2019all} emphasize Leo Breiman's {\em Rashomon effect}, the observation that there can often be many equivalently predictive models \citep{breiman2001statistical}.  
In light of this, \cite{rudin2022interpretable} urge machine learners to adopt models that directly lend themselves to interpretation, such as additive models.

Additive models enjoy the flexibility of nonlinear models while still being readily understandable to humans.  
By contrast, other nonlinear models that involve interactions between features are more difficult to fit and far harder to interpret, as they require explaining how one or more features modulate the effects of other features on the response.
Given the challenges of fitting and interpreting models with interactions, analysts should be cautious about moving beyond an additive model unless it is truly warranted. 
Motivated by this reasoning, \cite{yu2019reluctant} introduce a \textit{reluctance} principle that prioritizes the main effects in their modeling approach by including only those interactions that are highly correlated with the residuals after fitting a model with main effects and their squares.

In this paper, we seek a high standard of evidence for the existence of an interaction between two features, and in line with the reluctance principle described in \cite{yu2019reluctant}, we ask the following question: 
is there statistical evidence that a linear interaction between features is at play after having fit a SPAM to the data?  
We formulate this task as a hypothesis testing problem in which the SPAM fitted to our data acts as the null model.  
The p-values from our hypothesis testing approach guide us to a decision on whether to add specific interaction effects to the SPAM.

Taking such a testing approach, as natural as it may seem, falls outside the scope of classical statistical theory.
This is because the hypotheses being tested are data-adaptive in nature, involving a SPAM fit that is obtained only \textit{after} looking at the data.
Typically, a SPAM is fitted by solving a group lasso problem, where each group consists of the basis expansions for a feature, and the group lasso penalty promotes the selection of relevant main effects in the model.
A naive approach that ignores the data-adaptive nature of the tested hypotheses---using simple z-tests or t-tests for regression to compute p-values---may result in overly optimistic p-values. 
As a consequence, spurious interactions are more likely to be included in the model, even though they contribute no additional explanatory power beyond what is already captured by the SPAM fit, ultimately leading to a less accurate and a less interpretable model.

Our paper develops a selective inference method for constructing valid p-values for data-adaptive hypotheses on interaction effects, grounded in the principle of reluctance.
Unlike the naive p-values, our approach properly accounts for the SPAM fit and corrects for the over-optimism introduced by the data-adaptive process that comes prior to the test. 
More precisely, we derive valid p-values from the conditional distribution of test statistics, given the groups of main effects included in the SPAM fit. 
This conditional distribution is typically difficult to characterize and lacks a closed-form expression. 
To overcome this challenge, we make use of recently developed machinery in \cite{panigrahi2023approximate, huang2023selective} that add external Gaussian randomization to the group lasso objective.
The main highlights of our randomized approach are as follows:
\begin{enumerate}[leftmargin=0pt, labelsep=1em]
\item[(1)] The external randomization in our method facilitates tractable inference. 
By applying a change of variables to the Gaussian density of the randomization, we obtain a closed-form expression for the conditional distribution that accounts for the SPAM fit.
We present an efficient algorithm for computing approximate p-values and confidence intervals based on the conditional distribution.
\item[(2)] The added randomization generates a noisy version of the data, enabling a flexible lever to control how much information in our data is used in the first-stage SPAM fit.
This is similar to data splitting, another randomized method where the split proportion controls the amount of data used for selection versus inference. 
However, unlike data splitting, our p-values are based on statistics from the full dataset, rather than relying on just the holdout data used for inference in a splitting-based approach.

As a result, our p-values are more powerful than those from data splitting. 
Moreover, our method remains stable even as non-linear expansions of predictors and interactions create high correlations among features; by contrast, data splitting can fail to provide p-values in such scenarios, due to the smaller sample size of the holdout data which often leads to rank-deficient design matrices.
\item[(3)] In our numerical experiments, conducted across a wide range of settings, we employ Gaussian randomization variables with a small variance parameter. 
This ensures that the SPAM fit on the randomized (noisy) data closely matches the fit obtained using the full dataset. 
As a result, we are able to use nearly all the available data during both the main effects modeling stage and the interaction inference stage, thereby making better use of the available data.
We show that by taking the more rigorous selective inference approach, we consistently obtain a more accurate non-linear model than both the naive approach and the data splitting method.
\end{enumerate}

\subsection{Related work}
Before proceeding, we briefly review some work in three areas related to our contributions: nonlinear additive modeling, interaction modeling, and selective inference.

\noindent\textbf{Additive modeling}. \ Additive models offer the best of both worlds: they are powerful tools for nonlinear modeling while still producing fairly interpretable models. The work by \cite{ravikumar2009sparse} introduced the use of a group lasso-type of penalty or $\ell_2$-penalization to fit SPAMs to high-dimensional data (relatedly, \cite{lin2006component}).
\cite{lou2016sparse} extended this idea to high-dimensional settings by proposing a nested group lasso penalty, which allows for an adaptive choice of features to be treated as either linear or nonlinear.  Other extensions include \cite{chouldechova2015generalized,tay2020reluctant,haris2022generalized}.
When a small number of interactions are incorporated, additive models have been shown to achieve predictive accuracy on par with state-of-the-art nonlinear models, such as random forests, as demonstrated in the work of \cite{caruana2015intelligible}.

\noindent\textbf{Interaction modeling}. \ A standard approach for having main effects inform which interactions are included is through the hierarchical interaction assumption (also known as marginality and heredity) \cite{nelder1977reformulation, peixoto1987hierarchical}. Over the past few decades, several methods have formulated optimization-based approaches based on hierarchy, including \cite{bien2013lasso, yuan2009structured, haris2016convex,hazimeh2020learning}. 
%In \cite{yu2019reluctant}, a similar approach is proposed, where the inclusion of interactions is determined in a more flexible manner, as opposed to the hierarchical approaches that often rely on restrictive assumptions, such as perfect screening of main effects or even assumptions about the data type of features.
By contrast, \cite{yu2019reluctant} move away from hierarchy and instead adopt the reluctance principle.  In their approach,  interactions are considered for inclusion if they contribute to explaining the variability of the response in addition to an additive model.
These existing methods focus primarily on building a predictive model with interactions, with relatively few providing inferential guarantees about individual interaction effects.
In our work, we aim to rigorously implement the reluctant approach for including interactions to additive models by adopting a confirmatory framework based on hypothesis testing.

\noindent\textbf{Selective inference}. \ To test data-adaptive hypotheses, which are formed after examining the data, selective inference tools are essential. 
A selective inference approach to our problem, where the null hypothesis is based on the SPAM fit to the given data, involves obtaining p-values by conditioning on this initial SPAM fit before testing for interactions.
Early work along these lines, such as \cite{lee2016exact}, relied on the polyhedral nature  of the conditioning event of the lasso to construct inferential procedures. 
As a result, this approach does not extend to many other selection problems, including the selection of a SPAM in our case, where the conditioning event is no longer polyhedral in shape.

\cite{panigrahi2023approximate} introduced a conditional selective inference method after solving a group lasso problem with Gaussian data, which was later developed for more general models using a maximum likelihood approach in \cite{huang2023selective}.
This line of work proposed adding external randomization, in the form of Gaussian noise, to the group lasso optimization objective to obtain a tractable conditional distribution, without depending on the specific geometry of the selection event.
Furthermore, the randomized approach has the added benefit of producing more powerful inference compared to non-randomized methods, as shown for high-dimensional regression problems in \cite{panigrahi2023approximate, kivaranovic2024tight, panigrahi2024exact}.

\cite{panigrahi2023approximateML} proposed a maximum likelihood approach incorporating randomization for inference, which we adopt in this paper to perform inference on the interaction effects, given our initial SPAM fit.
New variants of data splitting have been proposed for Gaussian data in \cite{rasines2023splitting} and extended to other parametric distributions in \cite{leiner2023data,neufeld2024data,dharamshi2025generalized}.
However, like the widely used sample splitting approach, these methods condition on all the data used for selection, thereby conditioning on substantially more information than is necessary.
In a parallel line of work, simultaneous methods for selective inference, as pursued in \cite{berk2013valid, mccloskey2024hybrid, zrnic2024locally}, provide guarantees for the set of plausible data-adaptive targets that could have been formed, rather than for the observed event, such as the observed SPAM fit in our problem. 
Moreover, there is no existing work on how this approach could be applied to inference after the group lasso in an easily tractable manner.
Finally, we note that selective inference methods have been developed for non-linear regression models that do not involve additive modeling.
For example, \cite{neufeld2022tree, bakshi2024inference} develop selective inference for regression trees, while \cite{suzumura2017selective} develop methods for higher-order interaction models, leveraging the tree structure of various interaction orders.

The rest of this paper is organized as follows.
In Section \ref{sec:proposal}, we present the proposed method and illustrate its main features with a motivating data example.
In Section \ref{sec:method}, we derive the conditional distribution of key statistics, which yields p-values for testing the data-adaptive hypotheses based on the SPAM fit to given data.
In Section \ref{sec: sims}, we study the performance of our method under various settings, and in Section \ref{sec:application}, we apply our method to NYC flight data to predict flight departure delay. 
We conclude the work with Section \ref{sec:conclusion}. All proofs and additional data analysis results for Section \ref{sec:application} are collected in the appendix.

\setlength{\abovedisplayskip}{2pt}   

% Space below a full-width display when the following line is “normal” text.
\setlength{\belowdisplayskip}{2pt}   

% If the line before/after the display is short (e.g. end of a paragraph),
% these kicks in:
\setlength{\abovedisplayshortskip}{0pt}
\setlength{\belowdisplayshortskip}{0pt}

\section{The outline of our method} 
\label{sec:proposal}

Motivated by the reluctance principle in modeling interactions, we introduce a hypothesis testing method to determine if there is an interaction between features after fitting a SPAM. 
Our method, described in Algorithm \ref{algo:reluctant}, proceeds in two steps.

\noindent\textbf{Step 1:} 
In Step 1, we fit a SPAM by solving a randomized group lasso problem, which we display in \eqref{eq:grlasso} after first introducing some notation. Each feature $X_j\in\mathbb R^n$ is expanded using a $B_j$-dimensional spline basis to give $\Psi_j \in \mathbb{R}^{n \times B_j}$, leading to a design matrix $\Psi = [\Psi_1 \cdots \Psi_p] \in \mathbb{R}^{n \times q}$, whose $q = \sum_{j = 1}^p B_j$ columns are naturally partitioned into $p$ groups.  
%    Fixing some notations, let $\Psi_j \in \mathbb{R}^{n \times B_j}$ be  a $B_j$-dimensional spline basis expansion for the feature $X_j$. 
%We represent our design matrix as $\Psi = [\Psi_1 \cdots \Psi_p] \in \mathbb{R}^{n \times q}$, where $q = \sum_{j = 1}^p B_j$.
%This matrix contains the basis expansions for all our features.
%We include the column indices of $\Psi_j$ in $\Psi$ in the set $g_j$.

%For a fixed value $\tau^2$, we let $\omega \sim N({0}_q, \tau^2 I_q)$ be a randomization variable that is drawn independently of $y$.
For a user-specified positive definite matrix $\Omega$, let $\omega \sim N({0}_q, \Omega)$ be a randomization variable drawn independently of $y$.
To identify the main effects in our model, we solve:
\begin{align}
            \widehat{\beta}
            = 
            \argmin_{\left(\beta_{g_j} \in \real^{B_j}\right)_{j \in [p]}} 
            \frac 1{2}\left\| y-\sum_{j=1}^p\Psi_j\beta_{g_j}\right\|_2^2
            +\sum_{j=1}^p \left(\lambda_j \|\beta_{g_j}\|_2 + \frac{\epsilon}{2} \|\beta_{g_j}\|^2_2\right) 
            - \left(\beta_{g_1}^\top\ \dots\ \beta_{g_p}^\top\right)\omega,
            \label{eq:grlasso}
        \end{align}
where $\left\{\lambda_j\right\}_{j \in [p]}$ are positive tuning parameters, and $\epsilon > 0$ is a small ridge penalty parameter that guarantees strict convexity of the problem and thus the uniqueness of its solution.

Based on the solution 
$\widehat{\beta}= \begin{pmatrix} \widehat{\beta}_{g_1}^\top & \ldots & \widehat{\beta}^\top_{g_p}
            \end{pmatrix}^\top$,
let
\begin{align}
        \widehat{M}  = \left\{j \in[p]: \widehat{\beta}_{g_j}\neq {0}_{B_j}\right\}
        \label{eq:setM}
\end{align}
denote the set of selected groups (features). 
Finally, suppose $M$ is the observed value of $\widehat{M}$ for the given data and randomization $(y, \omega)$, which determines the additive main effects selected in the SPAM fit.
Let $q_M = \sum_{j \in M}B_j$, and let $\Psi_M \in \real^{n \times q_M}$ and $\Psi_{-M} \in \real^{n \times (q-q_M)}$ be a column concatenation of the matrices $\{\Psi_{g_j}: j \in M\}$ and $\{\Psi_{g_j}: j \notin M\}$, respectively.

\noindent\textbf{Step 2:} In Step 2, we perform a hypothesis test to determine if we must include linear interactions between features while taking into account the SPAM fit in Step 1. 
   
   For a pair $(j,k) \in [p]^2$, we denote the interaction between the features $X_j$ and $X_k$ by $\interjk = X_j \circ X_k \in \mathbb{R}^n$, where $\circ$ represents the component-wise product of two vectors.
   We use $\mathcal{T}^{M} \subseteq [p]^2$ to represent the indices of interactions that we want to test for inclusion in our model.
   Note that we allow this set to depend on data through the main effects model $M$ from the previous step, which is indicated by  superscript $M$.
   For example, under a weak hierarchy rule for including interactions, 
   $\mathcal{T}_{M} = \left\{(j,k) \in [p]^2: j\in M~\mathrm{or}~k\in  M \right\}$.

   In this step, we conduct a marginal screening of interactions in the set $\mathcal{T}_{M}$. 
   Specifically, for each $(j,k)\in \mathcal{T}^{M}$, we obtain a two-sided $p$-value for $H^{jk}_0: \tjk = 0$  in the model:
      \begin{equation}
        y \sim N(\Psi_{M} \bjk + \interjk \tjk, \sigma^2 I_n),
        \label{eq:postmodel}
    \end{equation}
   where $\bjk \in \real^{q_M}$ and $\tjk \in \real$ represent the main effect parameter vector and the interaction effect for $\interjk$, respectively.

\begin{algorithm}[]
\setstretch{1.35}
\label{algo:reluctant}
\caption{The proposed reluctant interaction inference framework}
    \begin{algorithmic}[100]
     \State \textbf{Step 1}: \\
     \hspace*{\algorithmicindent} \textbf{Require:} $y$: the response vector, $\Psi$: the design matrix using a spline expansion of features, $\omega$: the randomization vector\\
     
     \hspace*{\algorithmicindent} Fit a SPAM by solving \eqref{eq:grlasso}. \\
        
    \hspace*{\algorithmicindent} \textbf{Output:} 
    	$M$
      \State \textbf{Step 2}:\\ 
      %\hspace*{\algorithmicindent} \textbf{Input:} $\widehat{M} = M$\\
      \hspace*{\algorithmicindent} \textbf{Specify:} $\mathcal{T}^{M} \subseteq [p]^2$, a set of interaction indices that is allowed to depend on $M$
      
      \State\hspace{\algorithmicindent} %\textbf{Step 2}: 
      Using the model in \eqref{eq:postmodel}, obtain $p$-value associated with the null hypothesis $H^{jk}_0: \tjk = 0$ for each $(j, k) \in \mathcal{T}^{M}$.
      
      %\hspace*{\algorithmicindent} 
      \textbf{Output:} 
    	$\mathcal{S}^M \subseteq \mathcal{T}^M$: the set of interactions with significant p-values 
  \end{algorithmic}
\end{algorithm}

\subsection{A motivating data example}
\label{subsec:motivating_eg}

We simulate a data example to illustrate the inferential performance of our two-step method described in Algorithm
\ref{algo:reluctant}.
In particular, this example highlights the role of external randomization and demonstrates a trade-off between the quality of the first-step SPAM fit and the inferential power for recovering interactions in the second step. The data are generated from a nonlinear model, with details deferred to Section \ref{sec: sims}.

We compare inference for the interaction coefficient $\tjk$ using three different methods: (i) naive inference, which ignores the fact that a main-effects model was selected using the data, and proceeds with a $z$-test as if the SPAM were fixed in advance; (ii) data splitting, which partitions the data into two independent subsamples, one for selecting the SPAM and the other for inference on $\tjk$; and (iii) our proposed method using conditional inference. 

For our method, we implement randomization using $\Omega = \sigma^2 \frac{1-r}{r} \Psi^\top \Psi$, a natural choice for the randomization covariance matrix discussed further in Section \ref{subsec:sim_methods}. 
In the data splitting approach, $(r \times 100)\%$ of the data is used for selection. The parameter $r \in (0,1)$ determines the trade-off between data used for selection and data reserved for inference. 
While $r$ plays a similar role in our method, an important distinction is that we use the full data for inference rather than using only the holdout data as is done in data splitting.

Figure \ref{fig:ecdf_toy} shows the empirical CDFs of the pivots for inference on $\tjk$ constructed using the three methods. 
These pivots are obtained by applying a probability integral transform to the distribution of the test statistic used in each method; see Section~\ref{sec: sims} for their exact definitions. If the inference is valid, these pivots are expected to follow a uniform distribution.
While the pivots from naive inference deviate from the uniform distribution, as expected, the pivots from data splitting and the proposed method closely follow the uniform distribution, indicating the validity of inference with both approaches.

\begin{figure}[h]
    \centering
    \includegraphics[width=1.\linewidth]{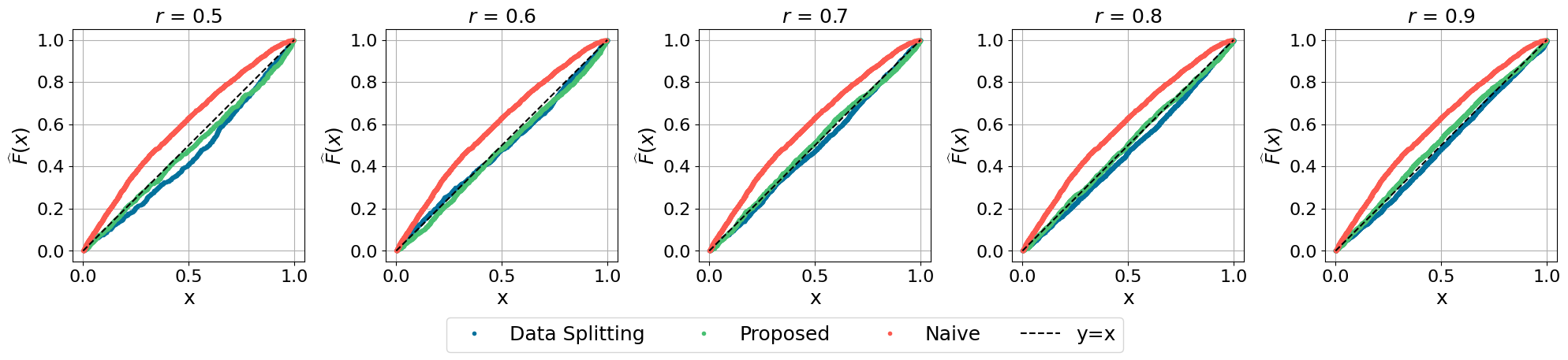}
    \caption{ECDF of the uniform pivots obtained by naive inference, data splitting, and the proposed method with varying amount of information used for selection ($r$)}
    \label{fig:ecdf_toy}
\end{figure}

In Figure \ref{fig:F1_len_toy}, we vary $r$ along the x-axis of the plots. The left panel displays the F1 scores, defined as the harmonic mean of precision and recall, which measure the quality of main effects recovery, or equivalently, the SPAM fit in Step 1. 
As $r$ approaches 1---meaning a larger fraction of the data is used for fitting the SPAM---we observe that the selection quality of both the proposed method and data splitting gradually converges to that of the naive method, which uses the full dataset for selecting main effects.

The right panel of Figure \ref{fig:F1_len_toy} shows that inference for $\tjk$ following the SPAM fit must account for the selection of the main effects. 
As expected, the average lengths of the confidence intervals produced by data splitting and our method widen as $r$ increases toward 1. 
However, the proposed method consistently yields shorter confidence intervals compared to data splitting, with the advantage in inferential power becoming even more pronounced for  larger values of $r$.

\begin{figure}[h]
    \centering
    \includegraphics[width=\linewidth]{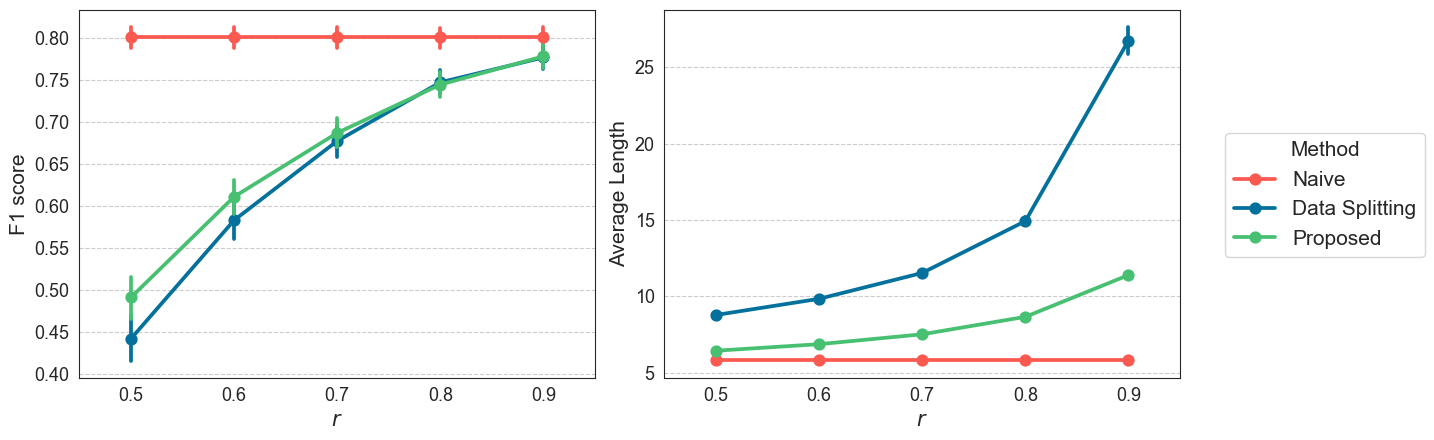}
    \caption{F1 scores for true main effects selection (left) and average confidence interval lengths for interaction coefficients (right) obtained by naive inference, data splitting, and the proposed method with varying amount of information used for selection ($r$)}
    \label{fig:F1_len_toy}
\end{figure}
Since we would like randomization to have minimal effect on the performance of Step 1, we fix $r=0.9$ for all following numerical experiments, at which point our randomized method closely matches the selection quality of the naive approach.  As such, our method enables the use of nearly the entire dataset for fitting the SPAM while still leveraging the full data for powerful inference on interactions.
These observations empirically support the main highlights (1), (2), and (3) of our method, outlined in Section \ref{sec:intro}.

\setlength{\abovedisplayskip}{2pt}   

% Space below a full-width display when the following line is “normal” text.
\setlength{\belowdisplayskip}{2pt}   

% If the line before/after the display is short (e.g. end of a paragraph),
% these kicks in:
\setlength{\abovedisplayshortskip}{0pt}
\setlength{\belowdisplayshortskip}{0pt}

\section{Inference procedures}
\label{sec:method}

In this section, we derive the p-values that we will use in Step 2.1 of Algorithm~\ref{algo:reluctant}.  Although $(y, \omega)$ have a simple joint normal distribution due to their independence, the conditional distribution of $(y, \omega)$ given the SPAM fit lacks a straightforward characterization.
To address this, we perform a change of variables on the known density of the randomization variables. This reduces the problem to characterizing the conditional distribution of $y$ and the randomized group lasso estimators. The change of variables transformation serves a dual purpose: (i) it yields a simple representation of the selection event as a set of sign constraints on the randomized group lasso estimator, and (ii) it enables a closed-form characterization of the joint conditional distribution of $y$ and the group lasso estimator.

We provide here a roadmap of what is to come.  In particular, in Section~\ref{sec:event}, we show that the KKT conditions, when suitably parameterized yield a simple characterization of the event that a certain set of features was selected by the SPAM fit in Step 1.  In Section~\ref{sec:statistics}, we describe the key statistics that will be involved in the p-value and derive their distribution if the set $M$ were fixed rather than selected based on the data.  In Section~\ref{sec:conditional}, we work out the exact conditional distribution of these key statistics given the event that SPAM selected a certain set of main effects.  In Section~\ref{sec:laplace}, we develop a Laplace approximation to the likelihood derived from this conditional distribution, making inference easier to implement.  Finally, in Section~\ref{sec:approxmle}, we show how to perform inference in this approximate likelihood.

\subsection{The conditioning event}
\label{sec:event}

To account for the selection of the SPAM prior to conducting tests for screening interactions, we condition on the event $\{\widehat{M} = M\}$.
Although this event does not have a simple description, there exists a subset of this event that can be characterized by a set of simple sign constraints on the group lasso estimator $\widehat{\beta}$.
Below, we review this characterization, which was identified in previous work by \cite{panigrahi2023approximate, huang2023selective}.

Hereafter, in the solution vector $\widehat{\beta}$, we assume without loss of generality that the nonzero groups where $\widehat{\beta}_{g_j} \neq 0$ are stacked above the zero groups where $\widehat{\beta}_{g_j} = 0$.
Consider the stationarity conditions of \eqref{eq:grlasso}, which implies
\begin{align}
\label{eq:KKTsubgradform}
	0_q &\in - \Psi^\top y - \omega + \left(\Psi^\top \Psi + \epsilon I_q \right) \begin{pmatrix}
		\left(\widehat{\beta}_{g_j}\right)_{j \in M}\\
		\left(0_{B_j}\right)_{j \notin M}
	\end{pmatrix}
 + 
\begin{pmatrix}
    \left(\lambda_j \dfrac{\widehat{\beta}_{g_j}}{\snorm{\widehat{\beta}_{g_j}}_2}\right)_{j \in M} \\
    \left(\lambda_j \partial\snorm{\widehat{\beta}_{g_j}}_2\right)_{j \notin M}
\end{pmatrix}.
\end{align}
Even though the stationarity conditions only need a subgradient to exist, the strictly convex quadratic loss function ensures that the subgradient is always unique. 
More precisely, for $j\notin M$, there exists a unique subgradient $\widehat{z}_{j} \in \partial\snorm{\widehat{\beta}_{g_j}}_2$ ($\iff \snorm{\widehat{z}_{j}}_2 \leq 1$) such that equality in \eqref{eq:KKTsubgradform} holds, i.e., 
\begin{comment}
\begin{equation*}
%\widehat{z}_j \in \partial\snorm{\widehat{\beta}_{g_j}}_2, \text{ and }
\snorm{\widehat{z}_{j}}_2 \leq 1,
\end{equation*}
\end{comment}
\begin{align}
\label{eq:KKT}
\omega &= - \Psi^\top y + \left(\Psi^\top \Psi + \epsilon I_q \right) \begin{pmatrix}
		\left(\widehat\gamma_j \widehat{u}_j \right)_{j \in M}\\
		\left(0_{B_j}\right)_{j \notin M}
	\end{pmatrix}
 + \begin{pmatrix}
    \left(\lambda_j \widehat{u}_j\right)_{j \in M} \\
    \left(\lambda_j \widehat{z}_j\right)_{j \notin M}
\end{pmatrix},
\end{align}
where, for $j\in M$,
$
 \widehat\gamma_j = \snorm{\widehat{\beta}_{g_j}}_2, \text{ and } \widehat\gamma_j >0,\;
 \ \widehat{u}_j = \frac{\widehat{\beta}_{g_j}}{\snorm{\widehat{\beta}_{g_j}}_2}.
$

%\fromjacob{In the previous line, I would write this as $\widehat{z}_j \in \partial\snorm{\widehat{\beta}_{g_j}}_2$ (rather than the equality), but actually for $j\notin M$ isn't this equivalent to $\snorm{\widehat{z}_{j}}_2 \leq 1$? So isn't it enough to just write $\snorm{\widehat{z}_{j}}_2 \leq 1$?}

%\fromjacob{If we want to treat the subgradient as a set then we might want to slightly modify the first display when the KKT conditions are presented.}
%\fromyiling{Thank you, Jacob! I have modified the KKT conditions a little bit for better rigor.}
Additionally, we define
\begin{align*}
    \widehat{{\gamma}} = (\widehat{\gamma}_j: j \in M)^\top \in \real^{|M|}, \quad \widehat{{U}} = \left( \widehat{u}_j: j \in M \right) \in \real^{q_M}, \quad
    \widehat{{Z}} = \left( \widehat{z}_j: j \notin M \right) \in \real^{q - q_M},
\end{align*}
as estimators that satisfy \eqref{eq:KKT}, where $\left(\widehat{u}_j: j \in M \right)$ denotes the stacking of $|M|$ vectors $\widehat{u}_j$'s, indexed by the set $M$, into one vector, and similarly for $\left( \widehat{z}_j: j \notin M \right)$.

We have written the KKT conditions in terms of $(\widehat{{\gamma}},\widehat{{U}},\widehat{{Z}})$ in place of $(\widehat{{\beta}},\widehat{{U}},\widehat{{Z}})$ because of the following lemma, which establishes a simplifying equivalence of events. 

%\fromhugo{Added dimensions for $\hat{U}$ and $\hat{Z}$, and removed the transpose notation for all $\hat{\theta}_{jk}^M$ to reflect the fact that $\hat{\theta}_{jk}^M$ is a scaler. Should we define $(\widehat{{\gamma}},\widehat{{U}},\widehat{{Z}})$ as a tuple or a concatenated vector as $(\widehat{{\gamma}}^\top,\widehat{{U}}^\top,\widehat{{Z}}^\top)^\top$? I modified the rest of the manuscript with the latter notation when a long vector is referred to. Will be happy to change it back.}

\begin{lem}[Equivalent characterization of model selection \citep{panigrahi2023approximate, huang2023selective}]
\label{lem:equiv}
Consider solving \eqref{eq:grlasso}.
Then, it holds that
\begingroup
  % make them very tight just around this one environment
  \setlength{\abovedisplayskip}{2pt}
  \setlength{\belowdisplayskip}{2pt}
  \setlength{\abovedisplayshortskip}{0pt}
  \setlength{\belowdisplayshortskip}{0pt}
\begin{align}
	\left\{\widehat{M} = M,\ \widehat{U} = U,\ \widehat{Z} = Z\right\}
  	= 
   \left\{\widehat{\gamma}\succ {0}_{|M|},\ \widehat{{U}} = U,\ \widehat{{Z}}\ = Z\right\}.
   \label{eq:equivalence}
\end{align}
\endgroup
\end{lem}

The event on the left-hand side of \eqref{eq:equivalence} is a proper subset of $\{\widehat{M}=M\}$, which is characterized by a set of positive constraints on the $\ell_2$-norms of the nonzero groups in the solution vector $\widehat{\beta}$.
Conditioning on this event takes into account the selection of main effects in the SPAM fit.
 %To interpret this result, note that if we fix the values of the sub-gradient variables $\widehat{{U}}$, $\widehat{{Z}}$ at their realized values $U$ and $Z$, the model selection event $\widehat{M} = M$ is alternatively represented by the affine event $\widehat{{\gamma}}\succ {0}_{|M|}$. The affine event $\widehat{{\gamma}}\succ {0}_{|M|}$ has a neat close-form expression, and will facilitate the integration of $\omega$ over the selection region. To take advantage of this convenient equivalence, we perform inference using estimators for the interaction coefficients and the main effects coefficients which we control for, referred to as the key statistics, conditional on the event \eqref{eq:equivalence}.
 
\subsection{The key statistics}
\label{sec:statistics}

In this section, we introduce what we call the ``key statistics'' that are involved in constructing a p-value for the null hypothesis $H_0^{jk}:\tjk=0$ in \eqref{eq:postmodel}.
In particular, our key statistics form a $(q+1)$-dimensional function of $y$, from which a p-value for testing $H_0^{jk}$ can be constructed based on their distribution conditional on ${\widehat{M} = M}$.

Denote by 
$\overline{Z}^{M}_{jk} = \begin{pmatrix}\interjk : \Psi_M
\end{pmatrix} \in \real^{n\times(q_M + 1)}$,
the interaction-augmented design matrix obtained by concatenating $\Psi_M$ with $\interjk$. 
Let
\begin{align}
\label{eq:lsmodel}
\begingroup
\renewcommand*{\arraystretch}{0.8}
    \begin{pmatrix}
    	\htjk\\
        \hbjk
    \end{pmatrix}
    \endgroup = \argmin_{\theta \in \real,\ \beta \in \real^{{q_M}}}  \left\|y - \Psi_M \beta - \interjk \theta \right\|_2^2 \ = \left(\overline{Z}^{M}_{jk}\right)^{+} {y},
\end{align}
be the least squares estimator when we perform a regression of $y$ against the derived features in $\overline{Z}^{M}_{jk}$, where $A^{+} = (A^\top A)^{-1}A^\top$ denotes the Moore-Penrose inverse of matrix $A$. %, where $\htjk \in \real$ and $\hbjk\in \real^{q_M}$.
However, the selection of the main effects not only depends on the least squares estimator based on the model \eqref{eq:postmodel}, but also on the residuals from this model fit and the features excluded from the SPAM fit. 
\begin{comment}
Therefore, the key statistics needed to perform the test also involve 
$$\perpjk = -\Psi_{-M}^\top \left(y - \overline{Z}^{M}_{jk}\begin{pmatrix}
\htjk\\
\hbjk
\end{pmatrix}\right) \in \real^{q-q_M},$$
which depends on $\Psi_{-M}$ and the residuals from the least squares model \eqref{eq:lsmodel}. \fromjacob{This appears to be the residuals not from the SPAM model but rather the residuals from the least squares that follows SPAM.} \fromyiling{Thanks for pointing out!}
\end{comment}
To see this, observe that $y$ can be written as the sum of the predicted value and the residuals from the fit as
\begin{align*}
    y = \overline{Z}^{M}_{jk}
    \begingroup
	\renewcommand*{\arraystretch}{0.8}
    \begin{pmatrix}
    	\htjk\\
        \hbjk
    \end{pmatrix} 
    \endgroup
    + \left(y - \overline{Z}^{M}_{jk}
    \begingroup
	\renewcommand*{\arraystretch}{0.8}
    \begin{pmatrix}
    	\htjk\\
        \hbjk
    \end{pmatrix}
    \endgroup
    \right).
\end{align*}
This in turn allows us to decompose $-\Psi^\top y$ involved in the KKT conditions in \eqref{eq:KKT} as
\begin{align}
\label{eq:in-terms-of-A}
    -\Psi^\top y = -\Psi^\top \overline{Z}^{M}_{jk}
    \begingroup
	\renewcommand*{\arraystretch}{0.8}
    \begin{pmatrix}
    	\htjk\\
        \hbjk
    \end{pmatrix} 
    \endgroup
    + \begingroup
	\renewcommand*{\arraystretch}{0.8}\begin{pmatrix}
        0_{q_M} \\
        \perpjk
    \end{pmatrix} \endgroup,
\end{align}
where $\perpjk = -\Psi_{-M}^\top \left(y - \overline{Z}^{M}_{jk}
\begingroup
\renewcommand*{\arraystretch}{0.8}\begin{pmatrix}
\htjk\\
\hbjk
\end{pmatrix}
\endgroup \right) \in \real^{q-q_M}$ 
involves the main effects excluded from the SPAM fit.  The $0_{q_M}$ follows from the normal equations of \eqref{eq:lsmodel}.

We denote the collection of these key statistics by
\begin{equation}\label{eq:key}
\Shat = \begin{pmatrix}
	\htjk, &
	\left(\hbjk\right)^\top, &
	\left(\perpjk\right)^\top
\end{pmatrix}^\top\in\real^{q+1}.
\end{equation}
We start by stating their distribution for a fixed $M$ (i.e., not conditioning on the event $\{\widehat{M}=M\}$) in Lemma \ref{lem:pre_sel_law}. Directly doing inference based on Lemma \ref{lem:pre_sel_law} but with $M$ replaced with $\widehat{M}$ would not properly account for the selection of main effects from SPAM and corresponds to what we will call throughout the naive approach. %, which does not account for the selection of the main effects in the SPAM fit. 

%To state our next result, we define the following matrices:
\begin{lem}[Marginal distribution of the key statistics]
\label{lem:pre_sel_law}
For any fixed set of additive main effects $M$ and any fixed pair of indices $(j,k)$,  consider the model 
$$y \sim N_n(\Psi_{M} \bjk + \interjk \tjk, \sigma^2 I_n).$$ 
In this model, %the density of $\Shat$ at $S \in \real^{q+1}$ is
%$$\phi\left(S; \nu^{M}_{jk}, \Sigma^{M}_{jk}\right),$$
$
\Shat\sim N_{q+1}\left(\nu^{M}_{jk},\Sigma^{M}_{jk}\right)
$
where 
$$\nu^{M}_{jk} = \begin{pmatrix}
    \tjk\\
    \bjk\\
    0_{q - q_M}
\end{pmatrix}
, \qquad \Sigma^{M}_{jk} = \begin{pmatrix}
	\overline{\Sigma}^{M}_{jk} & 0_{(q_M + 1) \times (q-q_M)}\\
	0_{(q - q_M) \times (q_M + 1)} & \widetilde\Sigma^M_{jk}
\end{pmatrix},$$ 
\begin{align*}
\overline{\Sigma}^{M}_{jk} &= \sigma^2\left(\left(\overline{Z}^{M}_{jk}\right)^\top \overline{Z}^{M}_{jk}\right)^{-1}\in\real^{(q_M+1)\times (q_M+1)},
\text{ and }\
\\
\ \widetilde\Sigma^M_{jk} &= \sigma^2 \Psi_{-M}^\top\left(I_q - \overline{Z}^{M}_{jk}\left(\overline{Z}^{M}_{jk}\right)^{+}\right)\Psi_{-M}\in\real^{(q-q_M)\times(q-q_M)}.   
\end{align*}
%and  $\phi\left(x; \mu_X, \Sigma_X 
%	\right)$ denotes the density of a multivariate Gaussian random variable with mean $\mu_X$ and covariance $\Sigma_X$ when evaluated at $x$. 
\end{lem}

Note that $\perpjk$ is an ancillary statistic for the parameters $\tjk$ and $\bjk$ in our model. 
A p-value for $\tjk$, when the set of main effects $M$ and the indices $(j,k)$ are fixed, can be obtained from the Gaussian distribution of 
$\left(\htjk, (\hbjk)^\top\right)^\top \mid \perpjk
\ \overset{d}{=}\ 
\left(\htjk, (\hbjk)^\top\right)^\top,
$
where the equality in distribution is due to independence between $\left(\htjk, (\hbjk)^\top\right)^\top$ and $\perpjk$. 

To account for the selection of $M$ using the same data, we adjust the naive distribution of the key statistics by using the conditioning event in Lemma \ref{lem:equiv}. 
In what follows, we derive a valid p-value using the conditional distribution
\begin{align}
\label{eq:condl_law}
	\left(\htjk, (\hbjk)^\top\right)^\top \Big\lvert \left\{\widehat{{\gamma}}\succ {0}_{|M|},\ \widehat{{U}} = U,\ \widehat{{Z}}\ = Z,\ \perpjk =\mathcal{A} \right\}.
\end{align}
After conditioning, $\perpjk$ is no longer an ancillary statistic, and we have the option to either integrate it out or condition on it. 
We do the latter to avoid calculating an integral over $q-q_M$ dimensions.
In the next section, we delve into the derivation of this distribution, which leads to a likelihood in $\left(\tjk,(\bjk)^\top \right)^\top$.

\subsection{Conditional distribution of the key statistics}
\label{sec:conditional}

%In this section, we derive the likelihood based on the conditional distribution in \eqref{eq:condl_law}. 
First, we fix some more notation and identify a mapping that allows us to compute the desired conditional distribution in \eqref{eq:condl_law}.

Let $\operatorname{bd}(\widehat{U}) = \operatorname{diag}(\widehat{u}_j: j \in M) \in \real^{q_M \times |M|}$ denote the block-diagonal matrix with $\widehat{u}_j$ on the diagonal blocks for $j \in M$.
Define
\begin{equation*}
\begin{gathered}
\mathcal{U}_M = \left\{u = \left({u}_j: j \in M \right) \in \real^{q_M}: u_j \in \mathbb{S}^{|g_j| - 1}\right\} \subseteq \real^{q_M},\\
\mathcal{Z}_M = \left\{z = \left( {z}_j: j \notin M \right) \in \real^{q-q_M}: \|z_j\|_2 \leq 1\right\}\subseteq \real^{q-q_M},
\end{gathered}
\end{equation*}
where $\left({v}_j: j \in M \right)$ denotes the stacking of $|M|$ vectors $v_j$, indexed by the set $M$, into one vector. Furthermore, denote by $\mathbb{S}^{k-1} = \{v \in \real^k: ||v||_2 = 1\}$ the unit sphere in $\real^{k}$.
Let $Q_M = \Psi_M^\top \Psi_{M}$ and let $\mathcal{{U}}_{\perp}^{(j)}(\widehat{u}_j) \in \mathbb{R}^{B_j \times (B_j - 1)}$ be an orthogonal basis completion of $\widehat{u}_j \in \mathbb{S}^{B_j - 1}$ for $j \in M$, and 
let
$$
\mathcal{U}_\perp(\widehat{U})=\operatorname{diag}\left(\left({\mathcal{U}}_{\perp}^{(j)}(\widehat{u}_j)\right)_{j \in M}\right),\ 
\Gamma(\widehat\gamma)=\operatorname{diag}\left(\left(\widehat{\gamma}_j I_{B_j-1}\right)_{j \in M}\right),\ 
\Lambda=\operatorname{diag}\left(\left(\lambda_j I_{B_j}\right)_{j \in M}\right).
$$ 

\begin{lem}
\label{lem:cov}
	Define the mapping $\Pi_{\Shat}(\cdot, \cdot, \cdot): \real^{|M|} \times \mathcal{U}_M \times \real^{q-q_M} \to \real^{q}$ by
 $$\Pi_{\Shat}(\gamma, {U}, {Z})= {A}^{M}_{jk}\begin{pmatrix}
			\htjk\\
			\hbjk
		\end{pmatrix} + {B}(U) {\gamma} + {c}(U, Z),$$
 where 
	\begin{align*}
		{A}^M_{jk} = -\Psi^\top \overline{Z}^{M}_{jk},
            \ B({U}) 
            =\left( \Psi^\top\Psi_M + \begin{pmatrix}
                 \epsilon I_{q_M} \\
                 0_{(q-q_M)\times q_M}
            \end{pmatrix} \right) \operatorname{bd}({U}), 
		\ 
		c(U, Z) = 
		\begin{pmatrix}
			{0}_{q_M}\\
			\perpjk
		\end{pmatrix}
			+ 
		\begin{pmatrix}
    		\left(\lambda_j {u}_j\right)_{j \in M} \\
    		\left(\lambda_j {z}_j\right)_{j \notin M}
		\end{pmatrix}.
	\end{align*}
 Consider the estimators $(\widehat{{\gamma}}, \widehat{{U}}, \widehat{{Z}})$ that satisfy the relationship in \eqref{eq:KKT}.
	Then, it holds that
 \begin{align*}
 (\widehat{{\gamma}}, \widehat{{U}}, \widehat{{Z}}) =  \Pi^{-1}_{\Shat}(\omega).
 \end{align*}
\end{lem}

As will be shown in the proof of Proposition \ref{prop:conditional:density}, Lemma \ref{lem:cov} provides us a tool for characterizing the distribution of the randomized group lasso statistics $(\widehat{{\gamma}},\ \widehat{{U}},\ \widehat{{Z}})$ through a change-of-variables applied to the known distribution of $\omega$.

\begin{prop}
Denote $\phi(x; a, B)$ as the density function of $N(a, B)$ evaluated at $x$. The density function of the conditional distribution of
\begingroup
  % make them very tight just around this one environment
  \setlength{\abovedisplayskip}{2pt}
  \setlength{\belowdisplayskip}{2pt}
  \setlength{\abovedisplayshortskip}{0pt}
  \setlength{\belowdisplayshortskip}{0pt}
$$\left(\Shat,\ \widehat{\gamma},\ \widehat{U},\ \widehat{Z} \right) \ \Big\lvert \ \left\{\widehat{{\gamma}}\succ {0}_{|M|}, \ \widehat{U} \in  \mathcal{U}_M, \ \widehat{Z} \in \mathcal{Z}_M\right\}$$
\endgroup
when evaluated at $(S, \gamma, U, Z)$ is proportional to
\begingroup
  % make them very tight just around this one environment
  \setlength{\abovedisplayskip}{2pt}
  \setlength{\belowdisplayskip}{2pt}
  \setlength{\abovedisplayshortskip}{0pt}
  \setlength{\belowdisplayshortskip}{0pt}
\begin{equation*}
%\label{cond:density:exp}
\begin{aligned}
\phi(S; \nu^{M}_{jk}, \Sigma^{M}_{jk}) \cdot \phi\left(\Pi_{S}({\gamma}, {U}, {Z});{0}_q, \Omega \right) & \cdot \det D_{\Pi_{S}}(\gamma, U, Z)\cdot 1_{\real^{|M|}_{+}}(\gamma) \cdot 1_{\mathcal{U}_M}(U) \cdot 1_{\mathcal{Z}_M}(Z)
\end{aligned}
\end{equation*}
\endgroup
where $D_{\Pi_{S}}(\gamma, U, Z) = \Gamma(\gamma) +  \left({\mathcal{U}}_\perp({U})\right)^\top Q^{-1}_M\Lambda \left({\mathcal{U}}_\perp({U})\right)$.
\label{prop:conditional:density}
\end{prop}

We are now ready to derive the likelihood, referred to as the selective likelihood, that accounts for the selection of the main effects during the SPAM fit.
We obtain an expression for this likelihood from the conditional density in Proposition \ref{prop:conditional:density}, by conditioning further on the observed values of $\widehat{U}$, $\widehat{Z}$ and $\perpjk$.
Before we do so, define 
\begingroup
  % make them very tight just around this one environment
  \setlength{\abovedisplayskip}{2pt}
  \setlength{\belowdisplayskip}{2pt}
  \setlength{\abovedisplayshortskip}{0pt}
  \setlength{\belowdisplayshortskip}{0pt}
\begin{align*}
\begin{gathered}
\overline{\Omega}= \left({B}(\widehat{{U}})^\top \Omega^{-1}{B}(\widehat{{U}}) \right)^{-1},
\quad
\overline{A}=-\overline{\Omega} {B}(\widehat{{U}})^\top \Omega^{-1} {A}^{M}_{jk},
\quad 
\overline{b}=-\overline{\Omega} {B}(\widehat{{U}})^\top \Omega^{-1} {c}(\widehat{{U}}, \widehat{{Z}}), \\
\overline{\Theta}=\left((\overline{\Sigma}^{M}_{jk})^{-1}-\overline{A}^{\top}\overline{\Omega}^{-1} \overline{A}+({A}^{M}_{jk})^{\top} \Omega^{-1} {A}^{M}_{jk}\right)^{-1},
\quad
\overline{R}=\overline{\Theta} (\overline{\Sigma}^{M}_{jk})^{-1},\\
\quad 
\overline{s}=\overline{\Theta}\left(\overline{A}^{\top}\overline{\Omega}^{-1} \overline{b}-
({A}^{M}_{jk})^{\top} \Omega^{-1} {c}(\widehat{{U}}, \widehat{{Z}})\right),
\end{gathered}
\end{align*}
\endgroup
%\fromhugo{Have we defined $\Omega$? The current setting considers a special case of $\Omega$, if I understand correctly.} \fromyiling{Thanks! This is fixed now.}
where $A_{jk}^M$, ${B}(\widehat{{U}})$, and ${c}(\widehat{{U}}, \widehat{{Z}})$ are as defined in Lemma \ref{lem:cov}.

Additionally, let
\begingroup
  % make them very tight just around this one environment
  \setlength{\abovedisplayskip}{2pt}
  \setlength{\belowdisplayskip}{2pt}
  \setlength{\abovedisplayshortskip}{0pt}
  \setlength{\belowdisplayshortskip}{0pt}
$$c(\tjk,\bjk) = \int \int_{g'\succ {0}_{|M|}}\phi\left(b'; \overline{R}\begingroup
	\renewcommand*{\arraystretch}{0.8}
	\begin{pmatrix}
    \tjk\\
    \bjk
	\end{pmatrix} 
	\endgroup+ \overline{s}, \overline{\Theta}\right) 
\cdot
\phi\left(g'; \overline{A}b' + \overline{b}, \overline{\Omega}\right)
	\cdot
	\det D_{\Pi_{S}}(g', U, Z) \ dg' \ db'.$$
\endgroup
 
\begin{thm}[Selective log-likelihood]
\label{thm:selective_likelihood}
The log-likelihood of  
$\left(\htjk, (\hbjk)^\top\right)^\top$ conditional on 
\begingroup
  % make them very tight just around this one environment
  \setlength{\abovedisplayskip}{2pt}
  \setlength{\belowdisplayskip}{2pt}
  \setlength{\abovedisplayshortskip}{0pt}
  \setlength{\belowdisplayshortskip}{0pt}
$$\left\{\widehat{{\gamma}}\succ {0}_{|M|},\ \widehat{{U}} = U,\ \widehat{{Z}}\ = Z,\ \perpjk =\mathcal{A} \right\},$$
\endgroup
is equal to
\begingroup
  % make them very tight just around this one environment
  \setlength{\abovedisplayskip}{2pt}
  \setlength{\belowdisplayskip}{2pt}
  \setlength{\abovedisplayshortskip}{0pt}
  \setlength{\belowdisplayshortskip}{0pt}
\begin{equation}
\begin{aligned}
\ell_M(\tjk,\ \bjk; U,Z,\mathcal{A}) 
= -\log c(\tjk,\bjk) + \log\phi\left(\begingroup
	\renewcommand*{\arraystretch}{0.8}
	\begin{pmatrix}
    \htjk\\
    \hbjk
	\end{pmatrix} 
	\endgroup; \overline{R}\begingroup
	\renewcommand*{\arraystretch}{0.8}
	\begin{pmatrix}
    \tjk\\
    \bjk
	\end{pmatrix} 
	\endgroup
+ \overline{s}, \overline{\Theta}\right).
	\label{eq:selectivelikelihood}
\end{aligned}
\end{equation}
\endgroup
\end{thm}

Theorem \ref{thm:selective_likelihood} provides us with an exact likelihood that accounts for the SPAM fit on the given data. 
This likelihood yields valid p-values for the hypothesis $H_0: \tjk=0$ for the pair of interactions between $X_j$ and $X_k$, where $(j,k) \in \mathcal{T}^{M}$.

Additionally, we note that this likelihood readily yields valid inference on the main effects $\hat\beta_{jk}^M$, although this is not our focus in this paper.
In the next section, we present a test statistic based on the maximum likelihood estimator for $\tjk$ in this likelihood and offer a computationally efficient algorithm for calculating a p-value for our hypothesis.

\subsection{A Laplace approximation for the selective likelihood}
\label{sec:laplace}

Following the line of recent research of approximate selective inference \citep{panigrahi2017mcmc, huang2023selective, panigrahi2023approximate}, we adopt a Laplace approximation to the normalizing constant $c(\tjk,\bjk)$, where 
$$c(\tjk,\bjk) \approx \sup_{g'\succ {0}_{|M|}, b'}\phi\left(b'; \overline{R}\begingroup
	\renewcommand*{\arraystretch}{0.8}
	\begin{pmatrix}
    \tjk\\
    \bjk
	\end{pmatrix} 
	\endgroup+ \overline{s}, \overline{\Theta}\right) 
\cdot
\phi\left(g'; \overline{A}b' + \overline{b}, \overline{\Omega}\right)
	\cdot
	\det D_{\Pi_{S}}(g', U, Z).$$
We refer the readers to \citep{huang2023selective} for formal probabilistic justifications for this approximation. 

In practice, we handle the constrained optimization through the addition to the objective of a {\em barrier function}, $$\text{Barr}(g)=\sum_{k=1}^{|M|} \log(1+1/(g_k)),$$ 
which imposes an increasingly severe penalty as the variable $g$ approaches the boundary of the positive orthant constraint.
In particular, we approximate $c(\tjk,\bjk)$ with $\widehat{c}(\tjk,\bjk)$ where
\begin{equation*}
\begin{aligned}
\log \widehat{c}(\tjk,\bjk) = & -\inf_{b', g'}
		\bigg\{ \frac{1}{2}\left(b' - \overline{R}
	\begingroup
	\renewcommand*{\arraystretch}{0.8}
	\begin{pmatrix}
    \tjk\\
    \bjk
	\end{pmatrix} 
	\endgroup
- \overline{s}\right)^\top\overline{\Theta}^{-1}\left( b' - \overline{R}\begingroup
	\renewcommand*{\arraystretch}{0.8}
	\begin{pmatrix}
    \tjk\\
    \bjk
	\end{pmatrix} 
	\endgroup 
	- \overline{s}\right)\\
		&\quad\quad\ + \frac{1}{2}\left(g' - \overline{A}b' - \overline{b}\right)^\top\overline{\Omega}^{-1}\left(g' - \overline{A} b' - \overline{b}\right)  - \log \det D_{\Pi_{S}}(g', U, Z)+{\text{\normalfont Barr}}( g')\bigg\}.
\end{aligned}
\end{equation*}

Dropping constants that do not depend on the parameters, we obtain our approximate selective log-likelihood:
\begin{equation}
\begin{aligned}
&\log \widehat{\mathcal{L}}_M(\tjk,\ \bjk; U,Z,s^M_{jk}) 
= \log \phi\left(\begingroup
	\renewcommand*{\arraystretch}{0.8}
	\begin{pmatrix}
    \htjk\\
    \hbjk
	\end{pmatrix} 
	\endgroup; \overline{R}
	\begingroup
	\renewcommand*{\arraystretch}{0.8}
	\begin{pmatrix}
    \tjk\\
    \bjk
	\end{pmatrix} 
	\endgroup
+ \overline{s}, \overline{\Theta}\right) - \log \widehat{c}(\tjk,\bjk).
	\label{eq:approxloglikelihood}
\end{aligned}
\end{equation}

\subsection{Maximum likelihood inference using the selective likelihood}
\label{sec:approxmle}

We now present a result that gives an expression for the MLE $(\widehat{\theta}^{jk}_{mle}, \widehat{\beta}^{jk}_{mle})$ of $(\tjk, \bjk)$, using the approximate selective log-likelihood given in \eqref{eq:approxloglikelihood}. The expressions of the MLE and the corresponding observed Fisher information matrix are obtained as the solution to a low $|M|$-dimensional convex optimization problem. 
Using these expressions, we conveniently construct p-values based on a Wald-type test statistic (and Wald-type confidence intervals) for inference on the interaction effects.

\begin{thm}
\label{MLE:EstimatingEqn}
Consider solving the $|M|$-dimensional optimization problem
$g^*(\htjk, \hbjk) = $
	\begin{equation}
 	\underset{g}{\arg\min}\ 
	\frac{1}{2}\left(g - \overline A
	\begingroup
	\renewcommand*{\arraystretch}{0.8}
	\begin{pmatrix}
    \htjk\\
    \hbjk
	\end{pmatrix} 
	\endgroup 
	- \overline b\right)^\top 
	\overline\Omega^{-1} 
	\left(g - \overline A
	\begingroup
	\renewcommand*{\arraystretch}{0.8}
	\begin{pmatrix}
    \htjk\\
    \hbjk
	\end{pmatrix} 
	\endgroup - \overline b\right)
	 - \log \det D_{\Pi_{S}}(g, U, Z)+{\text{\normalfont Barr}}( g).
\label{opt:esteqn}
\end{equation}
Then, the joint maximum likelihood estimator for $\begingroup
	\renewcommand*{\arraystretch}{0.8}
	\begin{pmatrix}
    \tjk\\
    \bjk
	\end{pmatrix} 
	\endgroup$ is obtained as
\begin{eqnarray*}
	\begingroup
	\renewcommand*{\arraystretch}{0.8}
	\begin{pmatrix}
    \widehat{\theta}^{jk}_{mle}\\
    \widehat{\beta}^{jk}_{mle}
	\end{pmatrix} 
	\endgroup 
	= 
	\overline{R}^{-1}
	\begingroup
	\renewcommand*{\arraystretch}{0.8}
	\begin{pmatrix}
    \htjk\\
    \hbjk
	\end{pmatrix} 
	\endgroup  
	- \overline{R}^{-1}\overline{s} + \overline{\Sigma}^{M}_{jk} \overline A^\top \overline\Omega^{-1}\left(\overline A\begingroup
	\renewcommand*{\arraystretch}{0.8}
	\begin{pmatrix}
    \htjk\\
    \hbjk
	\end{pmatrix} 
	\endgroup+ \overline b- g^*(\htjk, \hbjk)\right).
\end{eqnarray*}
%\end{theorem}
%
%\begin{theorem}
%\label{MLE:FInfo}
%	Consider solving the optimization problem in \eqref{opt:esteqn}.
	Let the $(|M|+1)\times(|M|+1)$ matrix
$	
K = $
$$\overline \Theta ^{-1} + \overline A^\top\overline\Omega^{-1} \overline A
	- \overline A^\top\overline\Omega^{-1} \big[\overline\Omega^{-1} 
	-\nabla^2 \log \det D_{\Pi_{S}}(g^*(\htjk, \hbjk), U, Z) 
	+ \nabla^2{\text{\normalfont Barr}}(g^*(\htjk, \hbjk))\big] ^{-1}\overline\Omega^{-1} \overline A.
$$
	The observed Fisher information matrix based on the approximate selective log-likelihood in \eqref{eq:approxloglikelihood} is equal to
$%	\begin{align*}
		I_{\text{mle}} = (\overline{\Sigma}^{M}_{jk})^{-1}K^{-1}(\overline{\Sigma}^{M}_{jk})^{-1}.
$%\end{align*}
\end{thm}

\begin{proof}
    The proof of this result is similar to that of Theorem 4.1 in \cite{huang2023selective} and is omitted for brevity.
\end{proof}

Using the expression for the MLE and the observed Fisher information matrix from Theorem \ref{MLE:EstimatingEqn}, we approximate the distribution of $\htjk$ using $N\left(\tjk, \sqrt{{I^{-1}_{\text{mle}}}[|M|+1,|M|+1]}\right)$, where ${I^{-1}_{\text{mle}}}[|M|+1,|M|+1]$ denotes the $((|M|+1), (|M|+1))$ entry in $I^{-1}_{\text{mle}}$. 
This leads to an approximate $\operatorname{Uniform}(0,1)$ pivot 
$\Phi(z_{jk})$, where $$z_{jk} = (\htjk - \tjk) / \sqrt{{I^{-1}_{\text{mle}}}[|M|+1,|M|+1]},$$
and an approximate $(1-\alpha)$-level confidence interval for the parameter $\tjk$, given by
$$
\widehat{\theta}_{mle}^{M} \pm z_{1-\alpha/2}\sqrt{{I^{-1}_{\text{mle}}}[|M|+1,|M|+1]}.
$$
%To interpret the confidence interval, when the interval covers 0, we conclude there is not evidence that the interaction term $\interjk$ contains information  unexplained after fitting the main effects $\Psi_E$, and we decide not to add the $(j,k)^{\text{th}}$ interaction into the SPAM model at the $1-\alpha$ level, and vice versa.

%\input{simulation.tex}
\setlength{\abovedisplayskip}{2pt}   

% Space below a full-width display when the following line is “normal” text.
\setlength{\belowdisplayskip}{2pt}   

% If the line before/after the display is short (e.g. end of a paragraph),
% these kicks in:
\setlength{\abovedisplayshortskip}{0pt}
\setlength{\belowdisplayshortskip}{0pt}

\section{Simulation}
\label{sec: sims}
First, we describe the simulation settings explored in this section, along with the metrics used to evaluate our method. We then present the key findings from these simulations.

\subsection{Simulation settings}
\label{subsec:sim_setting}

We simulate $n=200$ independent observations of $p=20$ features. For each observation $i$, the features are sampled as follows: we first draw $Z_{i,1}, ..., Z_{i,20}\sim N_{20}\left(0_{20}, \Sigma(\rho_1, \rho_2, \rho_{cross})\right)$, where
$$\Sigma(\rho_1, \rho_2, \rho_{cross}) = 
\begin{pmatrix}
    (1-\rho_1)I_3 + \rho_1 1_{3} 1_{3}^\top & \rho_{cross}1_{3} 1_{17}^\top\\
    \rho_{cross}1_{17} 1_{3}^\top  & (1-\rho_2)I_{17} + \rho_2 1_{17} 1_{17}^\top
\end{pmatrix}.
$$
We then transform the samples as $X_{i,j} = 2.5\cdot\Phi^{-1}(Z_{i,j})$ for $j=1,2,3$, and $X_{i, j} = \Phi^{-1}(Z_{i, j})$ for $j=4, ..., 20$, thereby generating features with varying effect strengths.
As a result,  marginally, $X_{i,1}, X_{i,2}, X_{i,3}\sim\mathrm{Uniform}(0,2.5)$, and $X_{i,4}, ..., X_{i,20}\sim \mathrm{Uniform}(0,1)$, while their joint distribution depends on the matrix $\Sigma(\rho_1, \rho_2, \rho_{cross})$ through the parameters $\rho_1$, $\rho_2$ and $\rho_{cross}$.  In particular, $\rho_1$ controls the dependence among signal features, $\rho_2$ controls the dependence among noise features, and $\rho_{cross}$ controls the dependence between these two groups of features. 

Next, we randomly sample a subset $\mathcal{G}$ of interaction pairs from the set $\{(j,k): j \neq k; \ j, k \in [20]\}$, with $|\mathcal{G}| = s_{inter}$. We then construct the mean response vector $\mu \in \mathbb{R}^n$, where the $i$-th entry is given by:
\begin{equation}
    \mu_i = \gamma_{main}\left[ 2\sin(2X_{i,1}) + X_{i,2}^2 + \exp(-X_{i,3}) \right] + \gamma_{inter}\sum_{(j,k) \in \mathcal{G}} X_{i,j} X_{i,k},
    \label{eq:sim_model}
\end{equation} 
and $\gamma_{main}, \gamma_{inter} \in \mathbb{R}$ are constants that determine the relative strength of the main and interaction signals. 
Finally, the response vector $Y \in \mathbb{R}^n$ is generated as: $Y = \mu + \epsilon$, where $\epsilon \sim N_n(0_n, \sigma^2 I_n)$.

We construct four main simulation settings, which we refer to as Settings 1--4.
\begin{enumerate}[leftmargin=0pt, labelsep=1em]
\item Setting 1: Varying noise level. \  We fix $s_{inter} = 5$, $\gamma_{main} = 2$, $\gamma_{inter} = 2$, $\rho_1 = \rho_2 = 0.6$, and $\rho_{cross} = 0.48$, and vary the standard deviation of the noise $\epsilon$ by setting $\sigma \in \{0.5, 1, 2, 4\}$.
\item Setting 2: Varying cross-correlation. \ We fix $s_{inter} = 5$, $\gamma_{main} = 2$, $\gamma_{inter} = 2$, $\rho_1 = \rho_2 = 0.6$, and $\sigma=2$, and vary the cross-correlation parameter in $\Sigma(\rho_1, \rho_2, \rho_{cross})$ by setting $\rho_{cross} \in \{0, 0.2, 0.4, 0.6\}$. 
\item Setting 3: Varying strength of interactions.\  We fix $s_{inter} = 5$, $\gamma_{main} = 2$, $\rho_1 = \rho_2 = 0.6$, $\rho_{cross} =0.48$, and $\sigma=2$, and vary $\gamma_{inter} \in \{0.5, 1, 2, 5\}$.
\item Setting 4: Varying the number of interactions.\  We fix $\gamma_{main} = 2$, $\gamma_{inter} = 2$, $\rho_1 = \rho_2 = 0.6$, $\rho_{cross} =0.48$, and $\sigma=2$, and vary $s_{inter} \in \{5, 10, 15, 20\}$.
\end{enumerate}
Furthermore, the motivating example in Section \ref{subsec:motivating_eg} is simulated using $s_{inter} = 5$, $\gamma_{main} = 2$, $\gamma_{inter} = 2$, $\rho_1 = \rho_2 = 0.6$, $\rho_{cross} = 0.48$, and $\sigma=2$.

\subsection{Methods and metrics}
\label{subsec:sim_methods}

We compare our proposed method with two alternative approaches: the naive approach and the data-splitting approach. 
All three methods for obtaining p-values for the interaction effects are described below. 
To model the non-linear main effects in all three methods, we construct a B-spline basis expansion $\Psi_j$ for each feature $X_j$, using degree 2 and 6 knots. 
This produces a design matrix $\Psi$ with dimensions $200 \times 40$. 
To allow for greater flexibility, one can increase the degree of the B-spline basis expansion. As demonstrated in our simulations results, we observe substantial advantages over the naive and the data-splitting approaches, even when using degree 2.

\noindent \textbf{Naive approach}. \ The first method we consider is a ``naive'' approach, where the entire observed dataset is used to obtain $\widehat{M}=M$ by solving \eqref{eq:grlasso} with $\omega = 0_q$, i.e., without randomization. 
Inference for $\tjk$ under the model \eqref{eq:postmodel} is conducted via the standard $z$-test. 
Specifically, we ignore the randomness in $\widehat{M}$ and the fact that it is estimated from the same data, proceeding to make inferences for $\tjk$ for each $(j, k) \in \mathcal{T}^{M}$ as if the realization $M$ were a fixed, predetermined set.
The corresponding Gaussian pivot for each $(j,k) \in \mathcal{T}^{M}$ is given by  $z^{naive}_{jk} = (\htjk - \tjk) / (\sigma_{jk})$, where $\sigma_{jk} =  \sigma\sqrt{\left((Z^{M+}_{jk})^\top Z^{M+}_{jk}\right)^{-1}_{1,1}}$. Note that if $M$ were a fixed set, then $z^{naive}_{jk} \sim N(0,1)$, and $\Phi(z^{naive}_{jk}) \sim \operatorname{Uniform}(0,1)$, where $\Phi(\cdot)$ denotes the cumulative distribution function (CDF) of a standard normal variable.

\noindent \textbf{Data splitting approach}. \ The second method we consider is the data splitting method, where we randomly partition the data into two subsamples: $[n] = \mathcal{S}_{sel} \cup \mathcal{S}_{inf}$, where $\mathcal{S}_{sel} \cap \mathcal{S}_{inf} = \emptyset$, with $\mathcal{S}_{sel}$ reserved for the selection of $\widehat{M}$ and $|\mathcal{S}_{sel}| = n_1 = \lfloor r*n \rfloor$, and $\mathcal{S}_{inf}$, with $|\mathcal{S}_{inf}| = n_2 = n - n_1$, reserved for inference of $\tjk$ for each $(j, k) \in \mathcal{T}^{M}$, after observing $\widehat{M}=M$. 
Specifically, denote $y^{sel}$ and $\interjk^{sel}$ as the entries of $y$ and $\interjk$ with indices in $\mathcal{S}_{sel}$, and $\Psi^{sel}$ as the submatrix of $\Psi$ with rows indexed by $\mathcal{S}_{sel}$. We define $y^{inf}$, $\Psi^{inf}$, and $\interjk^{inf}$ analogously. We then estimate $\widehat{M}$ by solving \eqref{eq:grlasso} with $\omega = 0_q$, replacing $y$ and $\Psi_j$ with $y^{sel}$ and $\Psi_j^{sel}$, respectively. Inference for $\tjk$ is performed using:
\begin{equation*}
        y^{inf} \sim N(\Psi_{M}^{inf} \bjk + \interjk^{inf} \tjk, \sigma^2 I_{n_2}).
\end{equation*}
The same procedure as in the naive approach is used to compute the pivots for the interaction effects, except that they are computed solely on $y^{inf}$, the data reserved for inference.

\noindent \textbf{The proposed approach}. \ We compute $M$ by solving \eqref{eq:grlasso}, and then calculate $\htjk$ according to \eqref{eq:lsmodel}. 
The corresponding approximate selective Gaussian pivot for each $(j,k) \in \mathcal{T}^{M}$ is given by $z^{mle}_{jk} = (\widehat{\theta}_{mle}^{M} - \tjk) / (\sigma_{jk})$, where $\sigma_{jk} = \sqrt{{I^{-1}_{\text{mle}}}[|M|+1,|M|+1]}$. Here, $\widehat{\theta}_{mle}^{M}$ is the selective MLE for $\tjk$ from Section~\ref{sec:approxmle} (with the subscript $jk$ suppressed), and ${I^{-1}_{\text{mle}}}[|M|+1,|M|+1]$ is its estimated variance.

In our method, we set the randomization covariance $\Omega = \sigma^2\frac{1-r}{r} \Psi^\top \Psi$, for some $r \in (0,1)$. 
As established in previous work \cite{panigrahi2021integrative, huang2023selective}, this choice of $\Omega$ leads to an asymptotically equivalent solution to data splitting, where $(100 \times r)\%$ of the data is used for selection, provided the main effects regression model is correctly specified. 
This equivalence provides a natural choice for setting $\Omega$, even when the main effects model is misspecified, i.e., when additional omitted interaction terms are present in the true data generation process.

To achieve main effects selection similar to naive inference, the proposed method is implemented using $\Omega = \sigma^2\frac{1-r}{r} \Psi^\top \Psi$ with $r = 0.9$, and data splitting is also performed using $90\%$ of the data for selection. 
For all methods, we fix the tuning parameter $\lambda_j = 0.5 \sigma \sqrt{n} \sqrt{B_j} \sqrt{2 \log(q)}$ in \eqref{eq:grlasso}, where $B_j = 2$ is the group size for the $j$-th variable. Furthermore, whenever the dispersion parameter $\sigma$ is not known, we replace it with the plug-in estimate $\widehat{\sigma} = \sqrt{||y - \Psi(\Psi^\top \Psi)^{-1}\Psi^\top y||_2^2 / (n-q)}$. Finally, for all methods, we construct $\mathcal{T}^{M}$ using the weak hierarchy rule, where $\mathcal{T}^{M} = \left\{(j,k) \in [p]^2: j \in M ~\mathrm{or}~k \in M \right\}$.

To assess the performance of our method and compare it with the other two approaches, we compute the following metrics:
\begin{enumerate}[leftmargin=0pt, labelsep=1em]
\item[(i)] Empirical cumulative distribution function (ECDF) of pivots. \ As noted before, if the inferential method is valid, with corresponding Type-I error rate control, then the ECDF of the pivots should approximately follow a $\operatorname{Uniform}(0,1)$ distribution. 
Therefore, for each inferential method, we compute pivots for the interaction effects, for each $(j,k) \in \mathcal{T}^M$, and plot the ECDF of these pivots, aggregated across all the interaction effects in   $\mathcal{T}^M$ and all simulation runs.
\item[(ii)] Lengths of confidence intervals. \ A common way to assess inferential power is by examining the lengths of the confidence intervals. For confidence intervals $\text{CI}(\tjk) = (L^{M}_{jk},\ U^{M}_{jk})$ corresponding to each $\tjk$, for $(j, k) \in \mathcal{T}^M$, we compute the average confidence interval length within each simulation, defined as
$$
    \text{Average Length} = \frac{\sum_{(j,k) \in \mathcal{T}^M} \left( U^{M}_{jk} - L^{M}_{jk} \right)}{|\mathcal{T}^M|},
$$
where the average is taken over all $(j,k) \in \mathcal{T}^M$. 
\item[(iii)] F1-Score evaluation of interaction effect recovery. \ Suppose we obtain a collection of p-values from each method, denoted by $\mathcal{P}^M = \left\{p_{jk}: (j,k) \in \mathcal{T}^M \right\}$. 
We define the test F1 score as
$$
\mathrm{F}1 \ \text{Score}(\mathcal{P}^M )=2\times \frac{\text {Precision}(\mathcal{P}^M) \times \text {Recall}(\mathcal{P}^M )}{\text {Precision}(\mathcal{P}^M)+ \text{Recall}(\mathcal{P}^M)},
$$
where the set of true discoveries consists of interaction effects whose magnitudes are greater than or equal to a predefined threshold $t_0 > 0$.
More specifically, to compute the F1 score, we calculate:
$$
\text {Precision}(\mathcal{P}^M)=\frac{\left|\left\{(j, k) \in \mathcal{I}^M: |\tjk| \geq t_0 \right\} \cap \left\{(j, k) \in \mathcal{I}^M: p_{jk} < \alpha \right\}\right|}{\left|\left\{(j, k) \in \mathcal{I}^M: p_{jk} < \alpha \right\}\right|};
$$
$$
\text {Recall}(\mathcal{P}^M)=\frac{\left|\left\{(j, k) \in \mathcal{I}^M: |\tjk| \geq t_0 \right\} \cap \left\{(j, k) \in \mathcal{I}^M: p_{jk} < \alpha \right\}\right|}{\left| \left\{(j, k) \in \mathcal{I}^M: |\tjk| \geq t_0 \right\} \right|}.
$$
In our simulations, we set $t_0=0.1$.
\end{enumerate}

\subsection{Results}

From Figure~\ref{fig:ECDF_sim}, we observe the following trends: 
(1) The proposed method produces valid pivots across all settings. Note that their empirical CDF closely aligns with the $\operatorname{Uniform}(0,1)$ distribution.
(2) The data splitting approach also yields valid pivots under most scenarios, similar to our proposed method. 
However, this approach fails to deliver valid inference in certain settings, such as Setting 3 with larger interaction effects or Setting 4 with a larger number of interaction effects. 
This breakdown in inferential validity is due to either rank deficiency in the design matrix or an insufficient sample size in the holdout data reserved for inference, highlighting situations where data splitting may not be a feasible strategy for selective inference.
(3) The naive approach consistently fails to produce valid inference across most settings. 
As expected, the empirical CDF of its pivots increasingly deviates from the uniform distribution as the noise level rises or as the collinearity between features increases, when selection bias due to main effect selection is much larger.

From Figure~\ref{fig:Len_sim}, we observe that the confidence intervals produced by our method are consistently shorter than those from the data splitting approach, which yields much wider intervals due to relying solely on holdout data for inference. 
In scenarios where the design matrix is rank-deficient or the holdout sample size is insufficient for reliable inference, the data-splitting method either yields extremely wide confidence intervals or fails to produce meaningful results. 
This issue can be noted for Settings 3 and 4, where the number of interaction effects is large or the magnitude of these effects is substantial.
Furthermore, our intervals are only slightly longer than those from the naive method, which does not account for the selection of the main effects. 
This highlights that our method provides valid interval estimates for the interaction effects without incurring a substantial increase in interval length compared to the invalid naive approach.

Finally, from Figure~\ref{fig:F1_sim}, we observe the following trends:
(1) The proposed method consistently achieves better or comparable F1-scores relative to both naive inference and data splitting, across the different settings, highlighting its strength in balancing false discovery rate (FDR) and power.
(2) The data-splitting method consistently struggles with insufficiency of data when testing null hypotheses after the selection of main effects. As noted in earlier plots, this results in situations where inference cannot even be computed, in which case a F1-score of 0 is assigned. The interactions identified by the data-splitting method result in models with significantly lower F1-scores in these settings.
(3) The only scenario where the proposed method yields a lower F1-score than naive inference is the special case of Setting 2 with $\rho = 0$, where all features are pairwise independent. 
In this case, naive inference produces valid pivots without requiring any selection adjustment. In all other scenarios, however, the proposed method outperforms naive inference, due its ability to provide valid inference while reusing the full data for inference.

\begin{figure}
    \centering
    \includegraphics[width=0.95\linewidth]{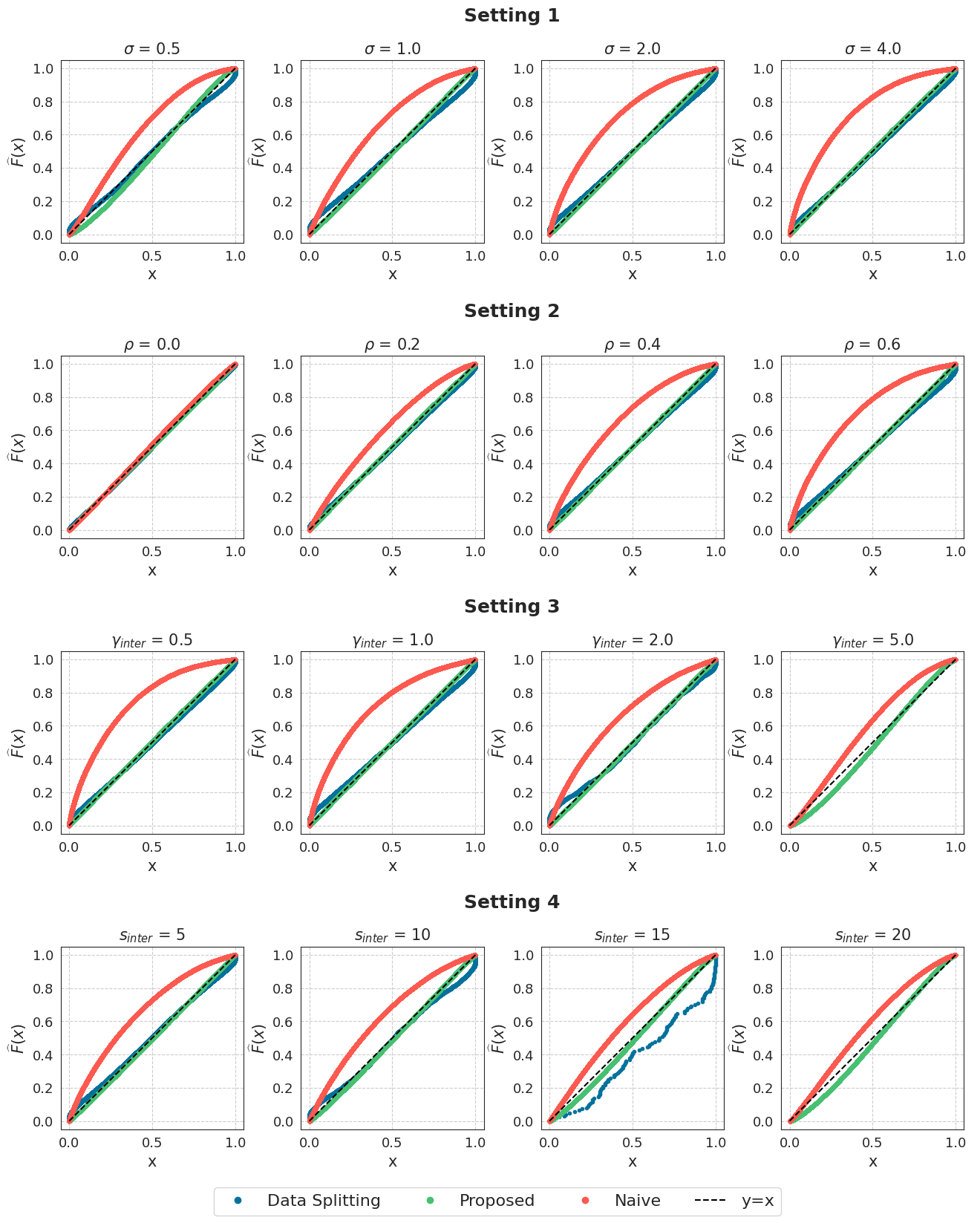}
    \caption{ECDF of the uniform pivots $\Phi(z_{jk})$ obtained by naive inference, data splitting, and the proposed method in all simulation settings}
    \label{fig:ECDF_sim}
\end{figure}

\begin{figure}
    \centering
    \includegraphics[width=\linewidth]{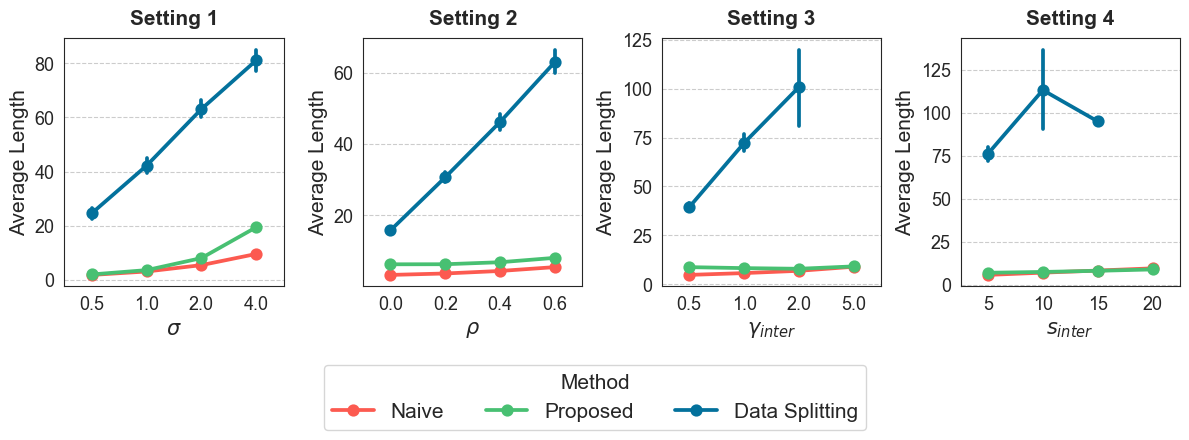}
    \caption{Average confidence interval lengths for $\tjk$ obtained by naive inference, data splitting, and the proposed method in all simulation settings}
    \label{fig:Len_sim}
\end{figure}

\begin{figure}
    \centering
    \includegraphics[width=0.98\linewidth]{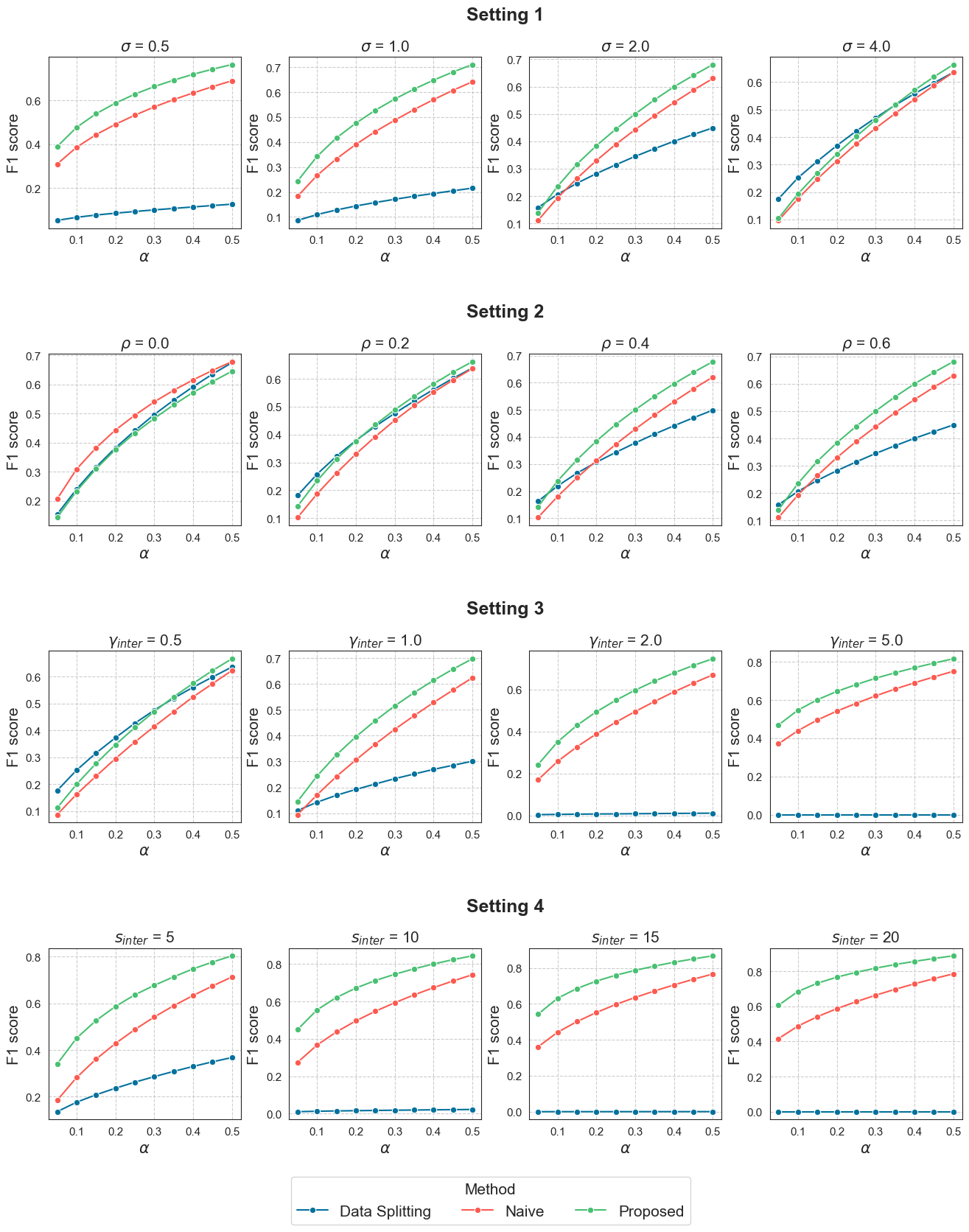}
    \caption{F1 scores of hypothesis testings of $\tjk$ obtained by naive inference, data splitting, and the proposed method in all simulation settings with varying $\alpha$}
    \label{fig:F1_sim}
\end{figure}

\section{Application: JFK flight data}
\label{sec:application}

In this section, we apply our method to a dataset consisting of flights information at the JFK airport in New York City during the time period from January 1, 2024 to January 30, 2024. 
The data can be accessed via the \texttt{R} package \texttt{anyflights} \cite{anyflights}. 
The dataset contains information on individual flights departing from JFK airport and includes three main components: (1) flight information, such as scheduled, departure and arrival times, destination, distance to destination, departure and arrival delays, and carrier details; (2) aircraft information, including the number of seats and the year of manufacture; and (3) weather information, such as temperature, precipitation, wind speed, and visibility, leading to a dataset of size $n \times p = 7358 \times 16$.

We fit our model to this data to predict arrival delay, measured in minutes, using flight, aircraft, and weather information that is available at the time of scheduled departure. 
Future information, such as actual departure and arrival time, which is not available at the time of prediction and thus not useful for this task, is excluded from our model fit.

\subsection{Reluctant modeling of interactions and evaluation}
\label{subsec:subsample}

To illustrate the accuracy of the reluctant model fitting approach with interactions, we first apply our method to a subsample containing about $10\%$ of the samples in the full data ($n = 7358$), resulting in a subsample of size $n_{\text{sub}} = 735$, and use the remaining data as a test set to evaluate this model.
% Then, on this subsample, we perform inference for the selected interaction terms in $\mathcal{T}^{M}$, using our proposed method and naive inference, described in Section \ref{subsubsec:naive}. 
For constructing high-quality nonlinear basis expansions, we model features with more than $40$ unique values nonlinearly, indexed by $\mathcal{N} \subset [p]$, and treat the remaining features linearly, indexed by $\mathcal{L} \subset [p]$. A summary of the linear and nonlinear features is provided in Table \ref{tab:l_nl_features} in Appendix \ref{subsec:app:l_nl}. 

For each feature $X_j$ modeled nonlinearly (i.e., $j \in \mathcal{N}$), we generate a B-spline basis expansion $\Psi_j$ of degree 2 with 6 knots. This results in a final design matrix $\Psi$ of dimension $735 \times 21$.
We implement the proposed method using $\Omega = \widehat\sigma^2 \times \frac{1 - r}{r} \Psi^\top \Psi$, with $r = 0.9$, as done in Section \ref{sec: sims}. Here, we use the plug-in estimate $\widehat{\sigma} = \sqrt{||y - \Psi(\Psi^\top \Psi)^{-1}\Psi^\top y||_2^2 / (n-q)}$ in place of the unknown dispersion parameter $\sigma$. For the selection of main effects, both the proposed method and the naive approach for interaction inference are implemented using the tuning parameter $\lambda_j = 0.5 \widehat\sigma \sqrt{n_{\text{sub}}} \sqrt{B_j} \sqrt{2 \log(q)}$, where the group size $B_j = 2$ for $j \in \mathcal{N}$ and $B_j = 1$ for $j \in \mathcal{L}$.

We now proceed to validate the discovery of significant $\tjk$'s identified in the $10\%$ subsample using the remaining $90\%$ of the data, which was held out during the modeling of main effects and interactions.
For each $(j,k) \in \mathcal{T}^M$, we fit an OLS model to the holdout $90\%$ of the dataset using the specification in \eqref{eq:postmodel}, and conduct a T-test for $\tjk$ on this holdout data. 
Given the large sample size of the test data, we treat the significant $\tjk$'s identified in the holdout set as a proxy for the true set of discoveries.
We then assess the quality of discoveries made on the $10\%$ subsample—using both the proposed method and naive inference—by comparing them to this ground truth in terms of precision, recall, and F1 score.
Put differently, we assess the degree of concordance between the findings from the heldout test data and those obtained from the $10\%$ subsample.

Figure~\ref{fig:main_effects} summarizes the selected nonlinear main effects identified by the proposed method and by naive inference. Notably, the added randomization term in the proposed method does not lead to a substantially different model fit. The non-randomized naive inference selects only one additional nonlinear main effect—scheduled arrival hour—which exhibits a relatively flat fitted additive effect on arrival delay, at zero.

\begin{figure}[h]
    \centering
    \includegraphics[width=\linewidth]{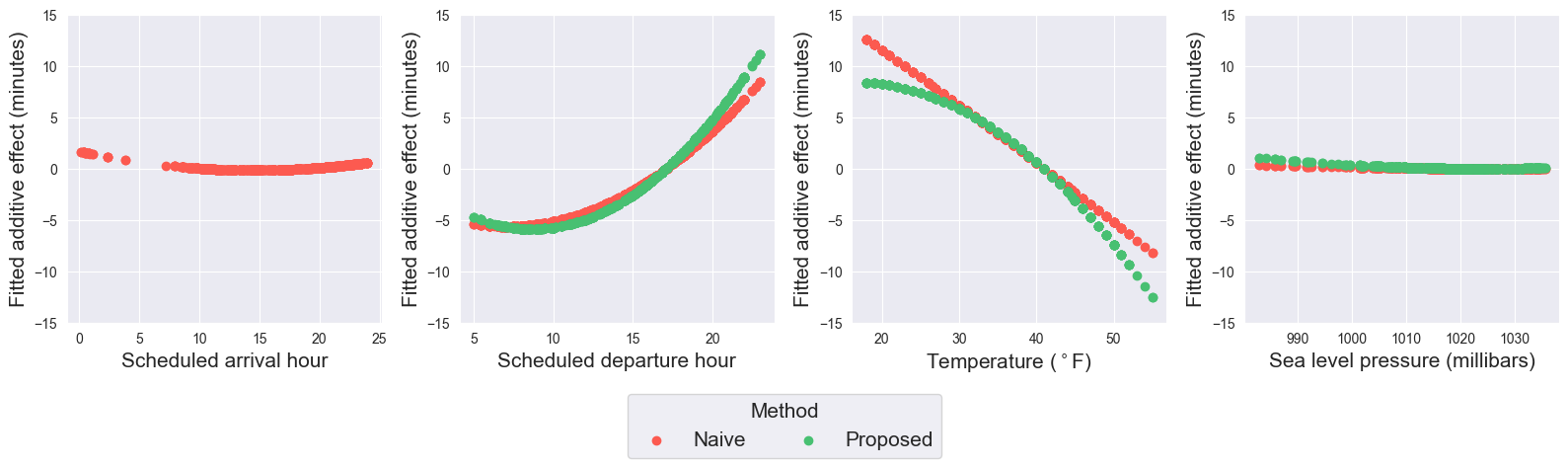}
    \caption{Fitted main effects using the proposed method and naive inference}
    \label{fig:main_effects}
\end{figure}

After fitting the main effects models, we construct $\mathcal{T}^{M}$ using the weak hierarchy rule, where $\mathcal{T}^{M} = \left\{(j,k) \in [p]^2: j \in M \vee k \in M \right\}$. Using a significance level of $\alpha = 0.1$, naive inference identifies 10 significant interaction terms, while the proposed method yields 13. Precision, recall, and F1 score based on the holdout validation set are reported in Table~\ref{tab:subsample}.

\begin{table}[h]
\centering
\begin{tabular}{@{}ccc@{}}
\toprule
          & Naive & Proposed       \\ \midrule
Precision & 0.4   & \textbf{0.615} \\
Recall    & 0.125 & \textbf{0.308} \\
F1 score  & 0.190 & \textbf{0.410} \\ \bottomrule
\end{tabular}
\label{tab:subsample}
\caption{Precison, recall, and F1 score of naive inference and the proposed method evaluated on the holdout sample}
\end{table}

As shown in Table~\ref{tab:subsample}, the proposed method outperforms naive inference across all three metrics. 
In the following section, we evaluate the replicability of these results by performing repeated subsampling, repeating the above analysis $500$ times.

\subsection{Replication analysis via repeated subsampling}

We repeatedly draw 500 subsamples of size %$10\%$ of the entire dataset, leading to subsamples of size
$n_{sub} = 735$ and in each repetition, we perform the analysis in Section \ref{subsec:subsample}.

\begin{figure}
    \centering
    \includegraphics[width=0.32\linewidth]{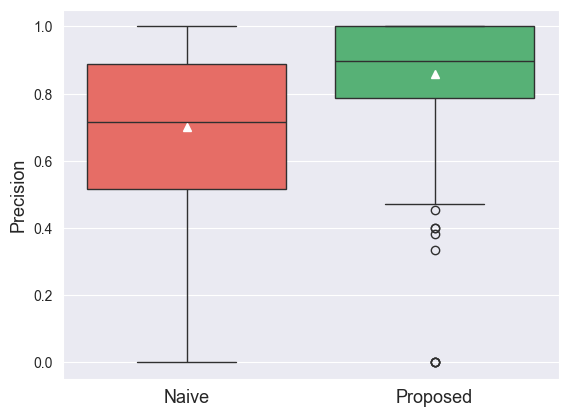}
    \includegraphics[width=0.32\linewidth]{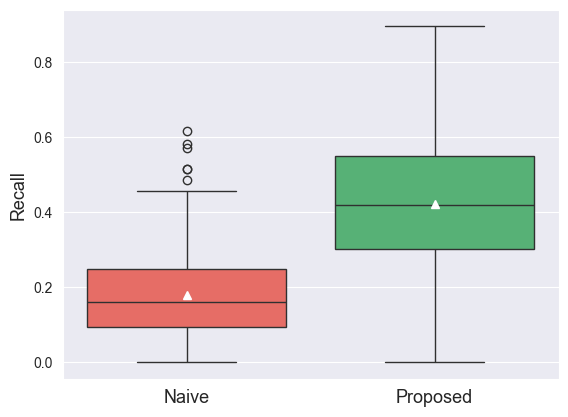}
    \includegraphics[width=0.32\linewidth]{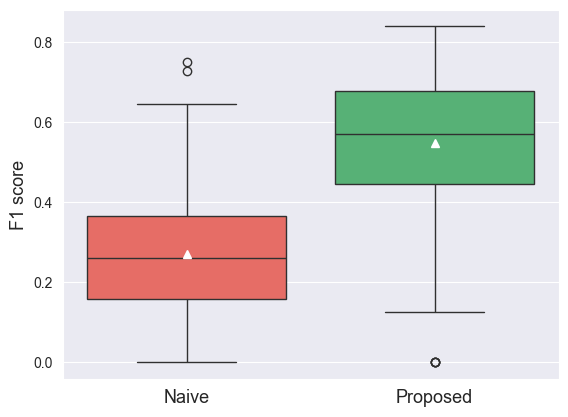}
    \caption{Precison, recall, and F1 scores of naive inference and the proposed method evaluated on the holdout sample, aggregated over 500 subsamples; empirical means are indicated by the white triangle}
    \label{fig:resampling}
\end{figure}

As shown in Figure~\ref{fig:resampling}, the proposed method consistently achieves higher precision, recall, and F1 score in recovering interaction effects, compared to the naive inference method. 
These findings are consistent with the results observed in the simulation studies. 

In Table~\ref{tab:freq_main}, we report the main effects that are most frequently selected.
Finally, in Figure \ref{fig:inter_freq}, we depict the frequency of significance of interaction signals that were most frequently deemed significant on the $500$ repetitions of the analysis.
This is based on the union of the top $30$ most frequently significant interaction signals from both naive inference and our approach. 

Note that the proposed method identifies the top interaction signals far more frequently than the naive approach. 
In particular, the interactions on the left half of Figure \ref{fig:inter_freq} are identified by the proposed method more than twice often, compared to naive inference. 
This further highlights the proposed method's ability to consistently recover key interaction signals, while the naive approach is likely susceptible to selection biases introduced during main effects selection, leading to an inferior model. A documentation of the features in Figure \ref{fig:inter_freq} is provided in Table \ref{tab:l_nl_features} in Appendix \ref{subsec:app:l_nl}.

\begin{table}[H]
\centering
\begin{tabular}{@{}ccc@{}}
\toprule
Feature                      & Naive & Proposed \\ \midrule
Sea level pressure (millibars) & 420   & 399      \\
Temperature ($^\circ$F)        & 400   & 324      \\
Visibility (miles)             & 370   & 357      \\
Scheduled departure hour       & 359   & 353      \\
Wind speed (mph)               & 348   & 381      \\ \bottomrule
\end{tabular}
\caption{Features that were selected for more than 300 times by both naive inference and data splitting, together with the number of times they are selected out of 500 replications}
\label{tab:freq_main}
\end{table}
\begin{figure}
    \centering
    \includegraphics[width=\linewidth]{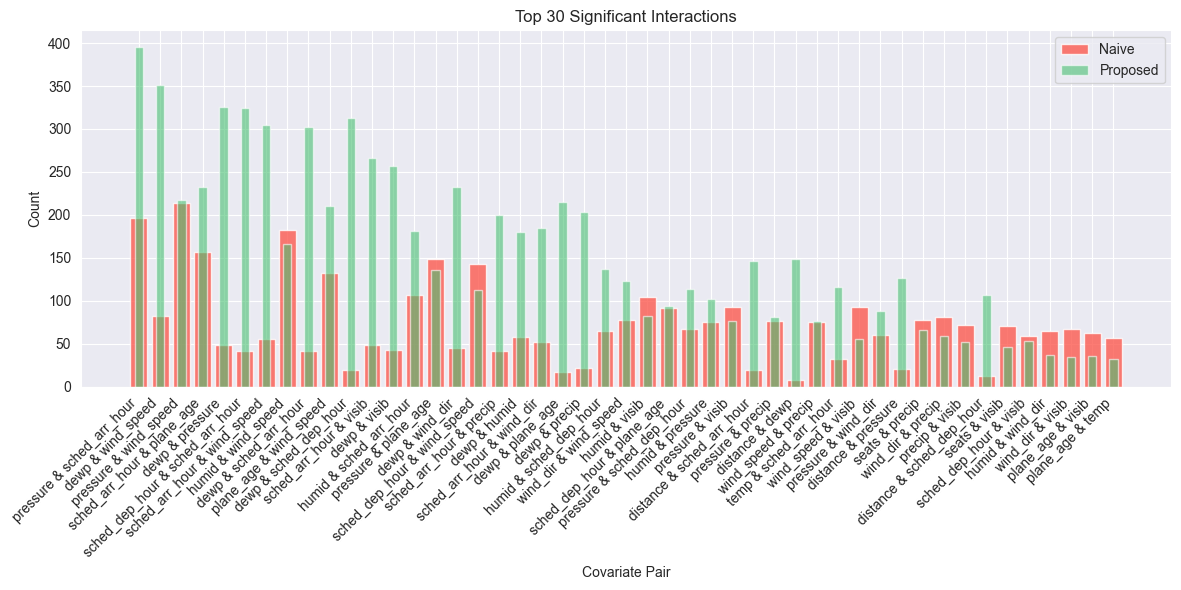}
    \caption{Frequency of significance of interaction signals that were most frequently deemed significant by naive inference the proposed method, aggregated over 500 subsamples; interactions are ordered along the x-axis according to the aggregated counts by the two methods.}
    \label{fig:inter_freq}
\end{figure}

\section{Conclusion}
\label{sec:conclusion}

In this paper, we develop a selective inference method to assess whether there is statistical evidence for moving beyond an additive, non-linear model. 
This approach follows the principle of reluctance, favoring simpler and more interpretable models without interactions unless it is truly warranted. 
We cast this as a data-adaptive hypothesis testing problem, where the null hypothesis assumes a SPAM that is fitted to the data prior to testing.
Our solution makes use of the recently developed selective inference techniques for selected groups after solving a group lasso problem, yielding valid p-values for the hypotheses on interaction effects.
Through extensive experiments on simulated and real data, we demonstrate that the resulting non-linear model---constructed via our rigorous, reluctant approach---achieves superior modeling accuracy compared to both naive strategies and existing selective inference methods such as data splitting.

While our primary focus in the paper is on testing for interactions, a natural by-product of our method is valid inference for the main effect variables included in the fitted SPAM, in the form of p-values and confidence intervals.
Another immediate extension of our method is to test groups of interaction terms jointly, rather than one at a time, using an F-test adjusted for the SPAM fit to compute the corresponding p-value.
 Finally, as a direction for future work, we note that a maximum likelihood approach could be developed for data-types beyond normal data by constructing an asymptotic selective likelihood, as developed in  \cite{huang2023selective, guglielmini2025selective}.

%%%%%%%%%%%%%%%%%%%%%%%%%%%%%%%%%%%%%%%%%%%%%%%%%%%%%%%%%%%%

\bibliographystyle{jasa3.bst}

\bibliography{references.bib}

\newpage
\appendix

\section{Proof of theoretical results}

\subsection{Marginal distribution of the key statistics}
\begin{proof}[Proof of Lemma \ref{lem:pre_sel_law}]
We first consider the sub-vector $\left(\htjk, (\hbjk)^\top\right)^\top$ in the first two components of $\Shat$. From \eqref{eq:lsmodel}, standard OLS theory establishes that 
$$
\begin{pmatrix}
    \htjk\\
    \hbjk
\end{pmatrix}
\sim N_{q_M + 1}\left(\begin{pmatrix}
    \tjk\\
    \bjk
\end{pmatrix}, \overline{\Sigma}^{M}_{jk}\right),
$$
under the model $y \sim N_n(\Psi_{M} \bjk + \interjk \tjk, \sigma^2 I_n).$ Then, 
\begin{align*}
    \perpjk =& -\Psi_{-M}^\top \left(y - \overline{Z}^{M}_{jk}
\begingroup
\renewcommand*{\arraystretch}{0.8}\begin{pmatrix}
\htjk\\
\hbjk
\end{pmatrix}
\endgroup \right)\\
=& -\Psi_{-M}^\top \left(I_n -  \overline{Z}^{M}_{jk}  \left(\overline{Z}^{M}_{jk}\right)^{+}\right) y \\
=& -\Psi_{-M}^\top \left(I_n -  \overline{Z}^{M}_{jk} \left(\overline{Z}^{M}_{jk}\right)^{+}\right) \left(\overline{Z}^{M}_{jk} \begin{pmatrix}
    \tjk\\
    \bjk
\end{pmatrix} + e\right)\\
=& -\Psi_{-M}^\top \left(I_n -  \overline{Z}^{M}_{jk} \left(\overline{Z}^{M}_{jk}\right)^{+}\right) e,
\end{align*}
where $e \sim N_n(0_n, \sigma^2 I_n)$, which implies 
$$\perpjk \sim N_n(0_{q-q_M}, \widetilde\Sigma^M_{jk}).$$
Then the result follows by establishing that 
$$\operatorname{Cov}\left(\begin{pmatrix}
    \htjk\\
    \hbjk
\end{pmatrix}, \perpjk \right) = 0_{q_M \times (q-q_M)},$$
but this follows from the orthogonality of the OLS coefficients and the residuals.    
\end{proof}

\subsection{Developing the conditional distribution of the key statistics}

\begin{proof}[Proof of Lemma \ref{lem:cov}]
Recall \eqref{eq:KKT}, which expresses the KKT conditions as
\begin{align}
\omega &= - \Psi^\top y + \left(\Psi^\top \Psi + \epsilon I_q \right) \begin{pmatrix}
		\left(\widehat\gamma_j \widehat{u}_j \right)_{j \in M}\\
		\left(0_{B_j}\right)_{j \notin M}
	\end{pmatrix}
 + \begin{pmatrix}
    \left(\lambda_j \widehat{u}_j\right)_{j \in M} \\
    \left(\lambda_j \widehat{z}_j\right)_{j \notin M}
\end{pmatrix}.
\end{align}
We will show that we can write the right-hand side in terms of our key statistics $\Shat$ instead of $y$. First, by \eqref{eq:in-terms-of-A},
$$
- \Psi^\top y = 
		{A}^{M}_{jk}
		\begin{pmatrix}
			\htjk\\
			\hbjk
		\end{pmatrix} + 
		\begin{pmatrix}
			{0}_{q_M}\\
			\perpjk
		\end{pmatrix}.
$$
Second, we observe that
 \begin{equation*}
 \begin{gathered}
 \left(\Psi^\top \Psi + \epsilon I_q \right) \begin{pmatrix}
		\left(\widehat\gamma_j \widehat{u}_j\right)_{j \in M}\\
		\left(0_{B_j}\right)_{j \notin M}
	\end{pmatrix} = {B}(\widehat{{U}}) \widehat{{\gamma}}.
 \end{gathered}
 \end{equation*}
 Thus, defining $c(U,Z)$ as in the statement of the lemma, we get 
 $$
 \omega =  \Pi_{\Shat}(\widehat{{\gamma}}, \widehat{{U}}, \widehat{{Z}}).
 $$
 Inverting the mapping $\Pi_{\Shat}$ proves our claim.
\end{proof}

\subsubsection{Proof of Proposition \ref{prop:conditional:density}}
\begin{proof}[Proof of Proposition \ref{prop:conditional:density}]
We derive the conditional distribution 
$$\left(\Shat,\ \widehat{\gamma},\ \widehat{U},\ \widehat{Z} \right) \ \Big\lvert \ \left\{\widehat{{\gamma}}\succ {0}_{|M|}, \ \widehat{U} \in  \mathcal{U}_M, \ \widehat{Z} \in \mathcal{Z}_M\right\},$$
by truncating the marginal density of $\left(\Shat,\ \widehat{\gamma},\ \widehat{U},\ \widehat{Z} \right)$ to the set $\real \times \real^{|M|}_{+} \times \mathcal{U}_M \times \mathcal{Z}_M$.
Note that the density function of this distribution, at $(S, \gamma, U, Z)$, is proportional to 
\begin{equation}
\label{cond:density:step1}
f(S) \cdot h_S(\gamma, U, Z) \cdot 1_{\real^{|M|}_{+}}(\gamma) \cdot 1_{\mathcal{U}_M}(U) \cdot 1_{\mathcal{Z}_M}(Z),
\end{equation}
where $f(\cdot)$ is the marginal density of the key statistics $\Shat$ and $h_S(\cdot)$ is the conditional density of 
$
\left(\widehat{{\gamma}},\ \widehat{{U}},\ \widehat{{Z}}\right) \mid \Shat=S.
$

Based on Lemma \ref{lem:pre_sel_law}, we already know that $f(S) \equiv \phi\left(S; \nu^{M}_{jk}, \Sigma^{M}_{jk}\right)$. 
In what follows, we determine $h_S(\gamma, U, Z)$ by using $\Pi_{\Shat}$, defined in Lemma \ref{lem:cov}, as a change of variables mapping from the randomization $\omega$ in \eqref{eq:grlasso} to $(\widehat{{\gamma}},\ \widehat{{U}},\ \widehat{{Z}})$, while conditioning on the value of $\Shat$.

Since $\Shat$ depends only on $y$ and not on the randomization variable $\omega$, which is independent of $y$, the density of $\omega \mid \Shat = S$ is equal to $\phi\left(\omega; {0}_q, \Omega \right)$.
Now, we apply the change of variable mapping 
    $$(\widehat{{\gamma}}, \widehat{{U}}, \widehat{{Z}}) =  \Pi^{-1}_{\Shat}(\omega)$$
where $\widehat{{\gamma}} \in \real^{|M|}_{+}$ and $\widehat{{U}} \in \mathcal{U}_M$.
At $\Shat=S$, this leads to 
$$h_S(\gamma, U, Z) \cdot 1_{\mathcal{U}_M}(U) = \phi\left(\Pi_{S}({\gamma}, {U}, {Z});
		{0}_q, 
            \Omega
		\right)
		\cdot \det D_{\Pi_{S}}(\gamma, U, Z) \cdot 1_{\mathcal{U}_M}(U)$$
 where  $D_{\Pi_{S}}(\gamma, U, Z)$ denotes the Jacobian matrix of the change of variable mapping $\Pi_{S }(\gamma, U, Z)$.
 Applying the same calculations as done in the proof of Theorem 2 in \cite{panigrahi2023approximate}, we obtain
$$\det D_{\Pi_S}(\gamma, U, Z)
  \propto \det\left(\Gamma(\gamma) + \left(\mathcal{U}_\perp(U)\right)^\top Q^{-1}_M\Lambda \mathcal{{U}}_\perp(U)\right).
$$
Thus, we conclude that the density in \eqref{cond:density:step1} is proportional to
\begin{align*}
\phi(S; \nu^{M}_{jk}, \Sigma^{M}_{jk}) \cdot \phi\left(\Pi_{S}({\gamma}, {U}, {Z});{0}_q, \Omega \right) & \cdot \det\left(\Gamma(\gamma) +  \left(\mathcal{U}_\perp(U)\right)^\top Q^{-1}_M\Lambda \mathcal{U}_\perp(U)\right)\\
&\;\;\;\;\;\;\;\;\;\;\;\;\;\;\;\;\;\;\;\;\;\; \times 1_{\real^{|M|}_{+}}(\gamma) \cdot 1_{\mathcal{U}_M}(U) \cdot 1_{\mathcal{Z}_M}(Z).
\end{align*}
\end{proof}

\subsubsection{Proof of Theorem \ref{thm:selective_likelihood}}
\begin{comment}
Furthermore, noticing that the marginal density of $\perpjk$ factors out in the above joint density, the following result presents the joint distribution of $(\htjk,\ \hbjk,\ \widehat{{\gamma}},\ \widehat{{U}},\ \widehat{{Z}})$, conditional on the value of $\perpjk$.

\begin{cor}
\label{lem:jointdist}
The distribution 
$$(\htjk,\ \hbjk,\ \widehat{{\gamma}},\ \widehat{{U}},\ \widehat{{Z}}) \mid \perpjk = a$$
has density function, evaluated at $\htjk = t,\ \hbjk=b,\ \widehat{{\gamma}}=\gamma,\ \widehat{{U}}=U,\ \widehat{{Z}}=Z$, proportional to
\begin{align*}
	\phi\left(\begingroup
	\renewcommand*{\arraystretch}{0.8}\begin{pmatrix}
		t\\
		b
	\end{pmatrix}\endgroup; 
	\begingroup
	\renewcommand*{\arraystretch}{0.8}
	\begin{pmatrix}
		\htjk\\
		\hbjk
	\end{pmatrix}
	\endgroup, \overline{\Sigma}^{M}_{jk}\right)
	\cdot
	\phi\left(\Pi_{s}({\gamma}, {U}, {Z});
		{0}_q, \Omega
		\right)
		\cdot J_{\Pi}({{\gamma}}, U,Z).
\end{align*}
\end{cor}
\end{comment}
We first state the following lemma.
\begin{lem}[Exact joint likelihood]
\label{lem:exact_likelihood}
Consider a density function that is proportional to
\begin{equation*}
    \phi\left(
    \begingroup
    \renewcommand*{\arraystretch}{0.8}
    \begin{pmatrix}
    t\\
    b
    \end{pmatrix}
    \endgroup; 
    \begingroup
    \renewcommand*{\arraystretch}{0.8}
    \begin{pmatrix}
    \tjk\\
    \bjk
    \end{pmatrix}\endgroup, \overline{\Sigma}^{M}_{jk}\right)  \cdot \phi\left(\Pi_{S}({g}, {U}, {Z});{0}_q, \Omega\right)  \cdot \det D_{\Pi_{S}}(g, U, Z),
\end{equation*} 
when evaluated at $(t, b, g)$. This density function admits an alternative representation and can also be expressed as being proportional to
\begin{align*}
	\phi\left(\begin{pmatrix}
    t\\
    b
\end{pmatrix}; \overline{R}\begin{pmatrix}
    \tjk\\
    \bjk
\end{pmatrix}+ \overline{s}, \overline{\Theta}\right) 
\cdot 
\phi\left(g; \overline{A}\begin{pmatrix}
    t\\
    b
\end{pmatrix}+ \overline{b}, \overline{\Omega}\right)
	\cdot
	\det D_{\Pi_{S}}(g, U,Z).
	\label{eq:jointdensity}
\end{align*}
\end{lem}
\begin{proof}
    The proof follows from straightforward algebra, by expanding the quadratic forms in the exponents of density functions $\phi\left(
    \begingroup
    \renewcommand*{\arraystretch}{0.8}
    \begin{pmatrix}
    t\\
    b
    \end{pmatrix}
    \endgroup; 
    \begingroup
    \renewcommand*{\arraystretch}{0.8}
    \begin{pmatrix}
    \tjk\\
    \bjk
    \end{pmatrix}\endgroup, \overline{\Sigma}^{M}_{jk}\right)$ and $\phi\left(\Pi_{S}({g}, {U}, {Z});{0}_q, \Omega\right)$, and reordering the terms. 
    A similar derivation can be found in the proof of Proposition 3.4 in \cite{huang2023selective}.
\end{proof}

Now, we are ready to state the proof of Theorem \ref{thm:selective_likelihood}.

\begin{proof}[Proof of Theorem \ref{thm:selective_likelihood}]
Recall from \eqref{eq:key} that $\Shat=\begin{pmatrix}
	\htjk, &
	\left(\hbjk\right)^\top, &
	\left(\perpjk\right)^\top
\end{pmatrix}^\top$.  Using Proposition \ref{prop:conditional:density}, and the independence between $\perpjk$ and $(\htjk, \hbjk)$, the conditional distribution 
$$\left(\htjk,\ \hbjk,\ \widehat{\gamma},\ \widehat{U},\ \widehat{Z} \right) \ \Big\lvert \ \left\{\widehat{{\gamma}}\succ {0}_{|M|}, \ \widehat{U} \in  \mathcal{U}_M, \ \widehat{Z} \in \mathcal{Z}_M,\ \perpjk = \mathcal{A}\right\}$$
evaluated at $(t, b, g, U, Z)$ has density
\begin{align*}
    \phi\left(
    \begingroup
    \renewcommand*{\arraystretch}{0.8}
    \begin{pmatrix}
    t\\
    b
    \end{pmatrix}
    \endgroup; 
    \begingroup
    \renewcommand*{\arraystretch}{0.8}
    \begin{pmatrix}
    \tjk\\
    \bjk
    \end{pmatrix}\endgroup, \overline{\Sigma}^{M}_{jk}\right)  \cdot \phi\left(\Pi_{S}({g}, {U}, {Z});{0}_q, \Omega\right) & \cdot \det D_{\Pi_{S}}(g, U, Z)  \cdot 1_{\real^{|M|}_{+}}(g) \cdot 1_{\mathcal{U}_M}(U) \cdot 1_{\mathcal{Z}_M}(Z).
\end{align*}
Further conditioning on $\{\widehat{{U}} = U\ \widehat{{Z}} = Z\}$,
%and applying Lemma \ref{lem:exact_likelihood} from this appendix, 
the density of $$(\htjk,\ \hbjk,\ \widehat{{\gamma}}) \mid \{\widehat{{\gamma}} \succ {0}_{|M|},\ \widehat{{U}} = U,\  \widehat{{Z}} = Z, \ \perpjk = \mathcal{A}\}$$
	evaluated at $(t, b, g)$ is equal to 
\begin{equation*}
    \dfrac{\phi\left(
    \begingroup
    \renewcommand*{\arraystretch}{0.8}
    \begin{pmatrix}
    t\\
    b
    \end{pmatrix}
    \endgroup; 
    \begingroup
    \renewcommand*{\arraystretch}{0.8}
    \begin{pmatrix}
    \tjk\\
    \bjk
    \end{pmatrix}\endgroup, \overline{\Sigma}^{M}_{jk}\right)  \cdot \phi\left(\Pi_{S}({g}, {U}, {Z});{0}_q, \Omega\right)  \cdot \det D_{\Pi_{S}}(g, U, Z)  \cdot 1_{\real^{|M|}_{+}}(g)}{\bigint \phi\left(
    \begingroup
    \renewcommand*{\arraystretch}{0.8}
    \begin{pmatrix}
    t'\\
    b'
    \end{pmatrix}
    \endgroup; 
    \begingroup
    \renewcommand*{\arraystretch}{0.8}
    \begin{pmatrix}
    \tjk\\
    \bjk
    \end{pmatrix}\endgroup, \overline{\Sigma}^{M}_{jk}\right)  \cdot \phi\left(\Pi_{S}({g'}, {U}, {Z});{0}_q, \Omega\right)  \cdot \det D_{\Pi_{S}}(g', U, Z)  \cdot 1_{\real^{|M|}_{+}}(g') dg' dt' db'}.
\end{equation*}
\medskip

Then, using Lemma~\ref{lem:exact_likelihood}, we observe that this density is further proportional to
\begin{align*}
	\phi\left(\begingroup
	\renewcommand*{\arraystretch}{0.8}
	\begin{pmatrix}
    t\\
    b
	\end{pmatrix}
	\endgroup; 
	\overline{R}\begingroup
	\renewcommand*{\arraystretch}{0.8}
	\begin{pmatrix}
    \tjk\\
    \bjk
	\end{pmatrix}
	\endgroup+ \overline{s}, \overline{\Theta}\right) 
\cdot 
\phi\left(g; 
	\overline{A}
	\begingroup
	\renewcommand*{\arraystretch}{0.8}
	\begin{pmatrix}
    t\\
    b
	\end{pmatrix}\endgroup + \overline{b}, \overline{\Omega}\right)
	\cdot
	\det D_{\Pi_{S}}(g, U, Z)
	\cdot
	1_{\real^{|M|}_{+}}(g).
	\label{eq:jointdensity}
\end{align*}

Marginalizing over $\widehat{{\gamma}}$ yields the joint likelihood
\begin{align*}
&c(\tjk,\bjk)^{-1} \times\\
&\phi\left(\begingroup
	\renewcommand*{\arraystretch}{0.8}
	\begin{pmatrix}
    \htjk\\
    \hbjk
	\end{pmatrix} 
	\endgroup; \overline{R}\begingroup
	\renewcommand*{\arraystretch}{0.8}
	\begin{pmatrix}
    \tjk\\
    \bjk
	\end{pmatrix} 
	\endgroup
+ \overline{s}, \overline{\Theta}\right) 
\cdot \int_{{{\gamma}} \succ {0}_{|M|}}
\phi\left({{\gamma}}; \overline{A}
\begingroup
	\renewcommand*{\arraystretch}{0.8}
	\begin{pmatrix}
    \htjk\\
    \hbjk
	\end{pmatrix} 
	\endgroup + \overline{b}, \overline{\Omega}\right)
	\cdot
	\det D_{\Pi_{S}}({{\gamma}}, U,Z) \ d{{\gamma}}.
\end{align*}
Dropping the integral that is independent of the parameters $\tjk, \bjk$ and taking the logarithm gives the desired result in \eqref{eq:selectivelikelihood}.
\end{proof}

\section{Additional details of the JFK flight data analysis}
\subsection{Documentation of features in Section \ref{sec:application}}
\label{subsec:app:l_nl}

\begin{table}[H]
\centering
\begin{tabular}{@{}cccc@{}}
\toprule
Feature name & Meaning & Source   & Type      \\ \midrule
distance & Distance between airports (miles) & Flight   & Nonlinear \\
sched\_dep\_hour & Scheduled departure hour          & Flight   & Nonlinear \\
sched\_arr\_hour & Scheduled arrival hour            & Flight   & Nonlinear \\
temp & Temperature ($^\circ$F)           & Weather  & Nonlinear \\
dewp & Dewpoint ($^\circ$F)              & Weather  & Nonlinear \\
humid & Relative humidity                 & Weather  & Nonlinear \\
pressure & Sea level pressure (millibars)    & Weather  & Nonlinear \\
wind\_dir & Wind direction (degrees)          & Weather  & Linear    \\
wind\_speed & Wind speed (mph)                  & Weather  & Linear    \\
precip & Precipitation (inches)            & Weather  & Linear    \\
visib & Visibility (miles)                & Weather  & Linear    \\
plane\_age & Plane age (miles)                 & Aircraft & Linear    \\
seats & Number seats on plane             & Aircraft & Linear    \\ \bottomrule
\end{tabular}
\label{tab:l_nl_features}
\caption{Feature names, meaning, source, and whether modeled linearly/nonlinearly in Section \ref{subsec:subsample}}
\end{table}

\end{document}